\documentclass[preprint,12pt]{elsarticle}

\usepackage{amsmath,amssymb,amsfonts,amsthm}
\usepackage{graphicx,booktabs,multirow}
\usepackage[hidelinks]{hyperref}
\graphicspath{{figures/}}
\usepackage[section]{placeins}
\usepackage{subcaption}
\usepackage{float}
\usepackage{tikz}

\newcommand{\qscale}[2]{%
\begin{tikzpicture}[x=11.5cm/#1,y=1.3cm]
\foreach \t in {0,#2,...,#1} {\draw[gray!40,densely dotted] (\t,0)--(\t,1);
  \draw[thick] (\t,-0.06)--(\t,0.06); \node[below,font=\scriptsize] at (\t,-0.06) {\t};}
\draw[gray!60,dashed] (0,1)--(#1,1);
\draw[very thick] (0,0)--(#1,0);
\node[left,font=\scriptsize] at (0,0) {0}; \node[left,font=\scriptsize] at (0,1) {1};
\node[font=\scriptsize] at (-0.06*#1,1.3) {$\mu$};
\end{tikzpicture}\\[2pt]
{\small L = \underline{\hspace{2cm}}\qquad C = \underline{\hspace{2cm}}\qquad U = \underline{\hspace{2cm}}}}

\newcommand{\D}{D^2_{\lambda\rho}}
\newcommand{\Tsc}{\mathcal{T}}
\newcommand{\Osc}{\mathcal{O}}

\newtheorem{definition}{Definition}
\newtheorem{proposition}{Proposition}
\newtheorem{remark}{Remark}
\newtheorem{assumption}{Assumption}

\journal{International Journal of Approximate Reasoning}

\begin{document}

\begin{frontmatter}

\title{Mediation and moderation in distance-based regression with triangular fuzzy data: simulation evidence and an application to academic performance}

\author[mercatorum]{Fabrizio Maturo\corref{cor1}}
\ead{fabrizio.maturo@unimercatorum.it}
\author[sanraffaele,irccs]{Annamaria Porreca}
\ead{annamaria.porreca@uniroma5.it}
\cortext[cor1]{Corresponding author.}

\affiliation[mercatorum]{organization={Department of Economics, Statistics and Business, Universitas Mercatorum},
            city={Rome}, country={Italy}}
\affiliation[sanraffaele]{organization={Department for the Promotion of Human Science and Quality of Life, San Raffaele University},
            city={Rome}, country={Italy}}
\affiliation[irccs]{organization={Unit of Clinical and Molecular Epidemiology, IRCCS San Raffaele Roma},
            city={Rome}, country={Italy}}

\begin{abstract}
We study mediation and moderation in the distance-based regression model for triangular fuzzy responses, in which the center and transformed spreads of the response are linear in the components of fuzzy and crisp covariates. With a common design, componentwise least squares minimizes the unconstrained quadratic extension of the Yang--Ko distance on the transformed scale, for any choice of its shape parameters. Direct, indirect and total effects are defined separately for the center and the spreads of the response as plug-in contrasts on the fitted equations; when the mediator spreads are transformed and carry part of the indirect association, the product-of-coefficients rule no longer holds for the center. Moderation is modeled through triangular interaction regressors and tested with a Freedman--Lane permutation test. In simulations calibrated on real data, percentile bootstrap intervals for the indirect effects had coverage between 0.905 and 0.975, the empirical size of the indirect-effect tests did not exceed 0.085, and ignoring the paths through the mediator spreads biased the indirect effect on the center by 16 to 20\% at every sample size. The fuzzy model predicted better than a centers-only model only from $n=250$. In data on 91 university students collected in 2014, the indirect contrast through a course indicator accounted for about half of the model-based association between prior knowledge and the center of academic performance, the interactions of prior knowledge with sex and with motivation remained after a Holm adjustment, and the spreads of the covariates did not improve leave-one-out prediction.
\end{abstract}



\end{frontmatter}

\section{Introduction}

Classical statistical modeling is traditionally grounded on the assumption that observations are precisely measured and that variables can be represented by single numerical values. In many real empirical contexts, however, this assumption is hardly tenable. Imprecision, vagueness and partial knowledge often affect the measurement process, the data collection phase, or even the conceptual definition of the variables under study. When such features are ignored and imprecise data are forcibly reduced to crisp values, a non negligible loss of information may occur, leading to models that provide a simplified and sometimes misleading representation of the underlying phenomena.

The mathematical foundations for dealing with vagueness were laid by \citet{Zadeh1965}, who introduced fuzzy sets as a generalization of classical set theory, while the analytic formalism of the operations on fuzzy sets was studied by \citet{Bellman1973}. Subsequent contributions by \citet{Zadeh1968} and \citet{Zadeh1975} further developed the formal apparatus of fuzzy logic, introducing fuzzy algorithms and linguistic variables as key tools for approximate reasoning. These works established fuzzy logic as a coherent mathematical framework for representing gradual concepts and imprecise information.

A conceptual distinction within this framework concerns the difference between fuzziness and probability. As emphasized by \citet{Kosko1990} and further discussed in \citet{Kosko1993}, probability theory models randomness and stochastic uncertainty, while fuzzy theory addresses vagueness as an intrinsic property of information. This distinction has been systematically analyzed within generalized information theory, see for instance \citet{KlirYuan1995} and \citet{Klir2006}. In this perspective, fuzziness is not interpreted as epistemic uncertainty but as structural imprecision, which persists even in the presence of complete information.

Within fuzzy set theory, fuzzy numbers represent imprecise numerical quantities. Comprehensive treatments of fuzzy sets and fuzzy numbers are provided in the classical monographs by \citet{DuboisPrade1980}, \citet{Zimmermann1996} and \citet{KlirYuan1995}. These contributions show how fuzzy numbers allow one to represent central tendency and dispersion simultaneously, preserving information that would otherwise be lost through crisp approximation. Alternative representations and operations for fuzzy numbers, including bell shaped fuzzy numbers associated with classical distributions, have been proposed and discussed in applied statistical contexts, see \citet{MaturoFortuna2016BellShapedFuzzyNumbers}.

The presence of fuzzy valued data motivates the extension of classical statistical models to accommodate imprecision. In the context of regression analysis, early contributions include fuzzy least squares approaches such as those proposed by \citet{Celmins1987} and \citet{Diamond1988}, which extend the least squares criterion to fuzzy vector data. These ideas were later formalized within a probabilistic framework through the concept of fuzzy random variables introduced by \citet{Puri1986}, and further developed in regression settings by \citet{Nather2006}. Inferential aspects for fuzzy random variables, including variance based testing procedures, are discussed in \citet{Ramos2010}.

A parallel line of research developed fuzzy regression models based on possibilistic semantics and inclusion constraints, as in \citet{Peters1994} and in the monograph by \citet{Tanaka1999}. Although these approaches played an important historical role, they interpret fuzziness mainly through possibility theory and model uncertainty, which is not always appropriate when fuzziness characterizes the observed data themselves rather than the underlying data generating process.

Distance-based fuzzy regression models depart from possibilistic interpretations: in them, fuzzy responses and predictors are treated as primitive data objects. In this framework, estimation is performed by minimizing suitable distance measures between observed and fitted fuzzy quantities, leading to a descriptive statistical interpretation of fuzziness. \citet{CoppiEtAl2006} proposed a least squares regression model for LR fuzzy responses based on the distance of \citet{YangKo1996}. Foundational contributions in this direction are provided by \citet{Blanco2009}, who propose estimation procedures for linear regression models with fuzzy random variables, and by \citet{Ferraro2010}, who introduce a linear regression model for imprecise responses in which the center and suitable invertible transformations of the spreads are modeled separately, so that fitted spreads are positive by construction. A determination coefficient for such models is subsequently developed in \citet{Ferraro2011}, and fuzzy explanatory variables are accommodated in the multiple regression model of \citet{FerraroGiordani2012}. Asymptotic properties of least squares estimators in regression models with vague concepts are investigated by \citet{Kratschmer2006a,Kratschmer2006b}. Comparative analyses between fuzzy and classical statistical regression models are discussed in \citet{Kim1996} and \citet{Wang2000}.

The present work belongs to the distance-based and descriptive tradition of fuzzy regression, in which fuzziness is interpreted as a structural property of the observed data rather than as epistemic uncertainty. This perspective has been adopted in economic and social sciences, where fuzzy regression models and alternative fuzzy operations have been proposed to address complex empirical phenomena involving vagueness, heterogeneity and partial observability, see for instance \citet{MaturoMaturo2013FuzzyRegressionCausality,MaturoHoskovaMayerova2017FuzzyRegression}. In the same descriptive spirit, fuzzy extensions have been proposed for other statistical tools, such as centrality measures for social networks in which the strength of the relations is a fuzzy quantity, with an application to a university department \citep{PorrecaMaturoVentre2025FuzzyCentrality}. A broader methodological discussion on the joint treatment of randomness and vagueness in applied sciences is provided in \citet{Maturo2016FuzzyProbabilisticApproach}.

Mediation and moderation are central in the analysis of relational structures \citep{BaronKenny1986,Hayes2018,VanderWeele2015}. Recent contributions have started to address them for fuzzy data. \citet{Yoon2020FuzzyModeration} introduced fuzzy moderation and moderated mediation models estimated by fuzzy least squares, with an $L_2$-type metric defined through support functions, fuzzy-valued coefficients and indirect effects obtained as products of path coefficients. \citet{ShinEtAl2026FuzzyMultipleModeration} extended this line to multiple moderators with $L_1$ and $L_2$ metrics, estimated by evolutionary algorithms and neural networks, and \citet{KangEtAl2026BootstrapFuzzyMediation} combined fuzzy least squares and fuzzy least absolute deviations with bootstrap inference for fuzzy mediation. These approaches rely on fuzzy arithmetic for the coefficients and estimate the model with optimization algorithms. They differ from the componentwise model of \citet{CoppiEtAl2006} and \citet{Ferraro2010}, in which the coefficients are real numbers, the spreads are modeled on a transformed scale and the estimator has a closed form. To our knowledge, mediation and moderation effects have not yet been formally defined within this componentwise model. In this model the product-of-coefficients rule for indirect effects no longer holds in general once the spreads are transformed.

The contribution of this paper is fivefold.
\begin{enumerate}
\item We show that, with a common design for the three component equations, componentwise least squares exactly minimizes the unconstrained quadratic extension of the Yang--Ko distance between observed and fitted fuzzy responses on the transformed scale, whatever its shape parameters (Proposition~\ref{prop:ls_yk}). This clarifies which distance the estimator minimizes and yields closed-form leave-one-out validation.
\item We define direct, indirect and total effects separately for the center and the spreads of the fuzzy response as plug-in contrasts on the fitted structural equations (Definition~\ref{def:effects}). We characterize when they reduce to the product-of-coefficients rule (Proposition~\ref{prop:decomp}) and treat fuzzy, crisp quantitative and crisp categorical predictors and mediators.
\item We define moderation through interaction regressors that remain triangular fuzzy numbers. We derive the corresponding conditional effects and propose a Freedman--Lane permutation test for the interaction block that uses the same distance as the estimator.
\item We assess finite-sample behavior through a Monte Carlo experiment calibrated on the empirical estimates and a plasmode experiment that resamples the observed covariates, reporting bias, root mean squared error, interval coverage, empirical size and power together with their Monte Carlo standard errors, and we isolate the bias of crisp analyses that omit some or all of the spreads.
\item We illustrate the methodology on data on university students collected in 2014 with a questionnaire based on triangular fuzzy assessments, and provide a single reproducible R pipeline.
\end{enumerate}

The paper is organized as follows. Section~\ref{sec:prelim} introduces fuzzy numbers, the operations used in the paper and the Yang--Ko distance. Section~\ref{sec:model} presents the regression model, the estimator and the tools for model assessment and inference. Section~\ref{sec:mediation} defines mediation effects, and Section~\ref{sec:moderation} defines moderation effects. Section~\ref{sec:simulation} reports the simulation studies, Section~\ref{sec:application} the empirical application, and Section~\ref{sec:discussion} discusses limitations and future research.

\section{Preliminaries on fuzzy sets and fuzzy numbers}
\label{sec:prelim}

Classical logic and classical set theory are based on the assumption that the membership of an element to a set is a binary property. An element either belongs to a set or it does not, and no intermediate states are allowed. While this framework is mathematically rigorous and computationally efficient, it proves inadequate for representing many concepts arising in real data analysis, where boundaries are often vague and gradual transitions are more realistic than sharp thresholds.

Fuzzy logic, introduced by \citet{Zadeh1965}, extends classical logic by allowing degrees of membership ranging continuously in the unit interval. This extension is particularly suitable for modeling linguistic concepts, subjective evaluations and imprecise measurements, which cannot be adequately captured by crisp categories. Fuzziness should not be confused with randomness: membership functions do not represent probabilities and do not satisfy the axioms of probability theory \citep{Kosko1990,Klir2006}. In the statistical framework adopted in this paper, the two notions coexist: the statistical units are randomly sampled, while the value observed on each unit is a fuzzy number \citep{Puri1986,Ferraro2010}.

\subsection{Fuzzy sets}

Let $\mathcal{X}$ be a non empty universe of discourse. A fuzzy set $A$ in $\mathcal{X}$ is defined by a membership function $\mu_A : \mathcal{X} \rightarrow [0,1]$, where $\mu_A(x)$ represents the degree to which $x \in \mathcal{X}$ belongs to $A$. When $\mu_A$ takes only the values $0$ and $1$, the fuzzy set reduces to a crisp set. Union, intersection and complement are commonly defined as $\mu_{A \cup B}(x) = \max\{\mu_A(x), \mu_B(x)\}$, $\mu_{A \cap B}(x) = \min\{\mu_A(x), \mu_B(x)\}$ and $\mu_{A^c}(x) = 1 - \mu_A(x)$.

\subsection{Triangular fuzzy numbers}
\label{subsec:tfn}

A fuzzy number $\tilde{A}$ is a fuzzy set on $\mathbb{R}$ whose membership function is normal, convex, upper semicontinuous and has bounded support \citep{DuboisPrade1980,Zimmermann1996}; its $\alpha$-cuts $\tilde{A}_\alpha=\{x:\mu_{\tilde{A}}(x)\ge\alpha\}$, $\alpha\in(0,1]$, are closed intervals.
An LR fuzzy number is written $\tilde{A}=(a_m,a_l,a_r)_{LR}$, where $a_m$ is the center and $a_l,a_r\ge 0$ are the left and right spreads, and $\mu_{\tilde{A}}(x)=L((a_m-x)/a_l)$ for $x\le a_m$ and $R((x-a_m)/a_r)$ for $x\ge a_m$, with $L$ and $R$ non-increasing shape functions such that $L(0)=R(0)=1$ \citep{DuboisPrade1980}. Triangular fuzzy numbers (TFNs) correspond to $L(z)=R(z)=\max\{0,1-z\}$. For $a_l,a_r>0$ the membership function is
\begin{equation}
\label{eq:tfn}
\mu_{\tilde{A}}(x)=
\begin{cases}
1-\dfrac{a_m-x}{a_l}, & a_m-a_l \le x \le a_m, \\[6pt]
1-\dfrac{x-a_m}{a_r}, & a_m \le x \le a_m+a_r, \\[4pt]
0, & \text{otherwise},
\end{cases}
\end{equation}
with $\alpha$-cuts
\begin{equation}
\label{eq:tfn_alpha}
\tilde{A}_\alpha=\bigl[a_m-(1-\alpha)a_l,\; a_m+(1-\alpha)a_r\bigr],\qquad \alpha\in(0,1].
\end{equation}
Since the spreads may be zero in the application, we take \eqref{eq:tfn_alpha} as the definition of a TFN: it is valid for all $a_l,a_r\ge0$, and a zero spread corresponds to a vertical side of the triangle, where \eqref{eq:tfn} is not defined. The center $a_m$ is the most plausible value, and the spreads quantify the imprecision on each side of it. The support is $[a_m-a_l,\,a_m+a_r]$. Throughout the paper, fuzzy quantities are written as triples (center, left spread, right spread). A crisp value $x$ is the degenerate TFN $(x,0,0)$.

Trapezoidal fuzzy numbers generalize TFNs by allowing a flat core interval, while bell shaped fuzzy numbers provide smoother membership functions. In the present work, TFNs are adopted for both theoretical developments and empirical analysis, as they offer a convenient balance between flexibility and analytical tractability.

\subsection{Operations and distance}
\label{subsec:operations}

For TFNs $\tilde{A}=(a_m,a_l,a_r)$ and $\tilde{B}=(b_m,b_l,b_r)$ and a real number $c$, the extension principle gives
\begin{equation}
\label{eq:sum_scalar}
\tilde{A}\oplus\tilde{B}=(a_m+b_m,\,a_l+b_l,\,a_r+b_r),
\qquad
c\odot\tilde{A}=
\begin{cases}
(c\,a_m,\,c\,a_l,\,c\,a_r), & c\ge 0,\\
(c\,a_m,\,|c|\,a_r,\,|c|\,a_l), & c<0,
\end{cases}
\end{equation}
so that sums and crisp multiples of TFNs are TFNs. The product of two TFNs is not triangular. For positive TFNs ($a_m-a_l\ge 0$, $b_m-b_l\ge 0$), the $\alpha$-cuts of the exact product are
\[
\bigl[(a_m-(1-\alpha)a_l)(b_m-(1-\alpha)b_l),\;(a_m+(1-\alpha)a_r)(b_m+(1-\alpha)b_r)\bigr],
\]
whose endpoints are quadratic in $\alpha$. Dropping the quadratic terms gives the classical triangular approximation \citep{DuboisPrade1980}
\begin{equation}
\label{eq:dp_product}
\tilde{A}\otimes\tilde{B}\approx(a_m b_m,\; a_m b_l+b_m a_l,\; a_m b_r+b_m a_r),
\end{equation}
which preserves the core exactly. Relative to the exact product, the approximation overstates the left spread by $a_l b_l$ and understates the right spread by $a_r b_r$, so that both endpoints of its support lie below the exact ones by these amounts. In particular, the approximate product of two positive TFNs can have a partly negative support: for $\tilde{A}=\tilde{B}=(1,0.8,0.8)$ it is $(1,1.6,1.6)$, with support $[-0.6,2.6]$, whereas the exact support is $[0.04,3.24]$.

Following \citet{CoppiEtAl2006} and \citet{Ferraro2010}, the discrepancy between LR fuzzy numbers is measured by the distance of \citet{YangKo1996},
\begin{equation}
\label{eq:yk}
\D(\tilde{A},\tilde{B})=(a_m-b_m)^2+\bigl[(a_m-\lambda a_l)-(b_m-\lambda b_l)\bigr]^2+\bigl[(a_m+\rho a_r)-(b_m+\rho b_r)\bigr]^2,
\end{equation}
where $\lambda=\int_0^1 L^{-1}(\omega)\,d\omega$ and $\rho=\int_0^1 R^{-1}(\omega)\,d\omega$ account for the shape of the membership function. For TFNs, $\lambda=\rho=1/2$. With $\lambda=\rho=1$, \eqref{eq:yk} compares the three vertices of the triangles, as in the fuzzy least squares criterion of \citet{Diamond1988}.

Distance based measures provide a natural way to quantify the discrepancy between observed and fitted fuzzy quantities and form the basis for fuzzy least squares estimation. Within this framework, fuzzy numbers are treated as data objects endowed with a metric structure, allowing the extension of classical statistical techniques to settings characterized by imprecision.

\section{Distance-based fuzzy regression model}
\label{sec:model}

Regression analysis represents one of the most widely used tools in statistical modeling for studying the relationship between a response variable and one or more explanatory variables. In many empirical contexts, measurements may be affected by imprecision, aggregation, subjectivity or linguistic assessment, leading to observations that are more appropriately represented by fuzzy numbers rather than single values. In this work, we focus on fuzzy regression models in which fuzziness is treated as a structural property of the observed variables, rather than as uncertainty on the model parameters, consistently with the distance based framework of \citet{CoppiEtAl2006}, \citet{Blanco2009}, \citet{Ferraro2010} and \citet{Ferraro2011}.

\subsection{Model specification}
\label{subsec:model}

Let $\tilde{Y}_i=(Y_{im},Y_{il},Y_{ir})$, $i=1,\ldots,n$, be a triangular fuzzy response with $Y_{il},Y_{ir}\ge 0$. Let $\tilde{X}_{ti}=(X_{tim},X_{til},X_{tir})$, $t=1,\ldots,q$, be fuzzy covariates and $T_{1i},\ldots,T_{si}$ crisp covariates. Collect the intercept, the centers and spreads of the fuzzy covariates and the crisp covariates in the row vector
\begin{equation}
\label{eq:design_row}
\begin{aligned}
\mathbf{x}_i^\top=\bigl(&1,\;X_{1im},X_{1il},X_{1ir},\;\ldots,\;X_{qim},X_{qil},X_{qir},\\
&T_{1i},\ldots,T_{si}\bigr)\in\mathbb{R}^{p},\qquad p=1+3q+s,
\end{aligned}
\end{equation}
and let $\mathbf{X}$ be the $n\times p$ design matrix with rows $\mathbf{x}_i^\top$. The model is
\begin{equation}
\label{eq:model}
Y_{im}=\mathbf{x}_i^\top\mathbf{a}_m+\varepsilon_{im},\qquad
g(Y_{il})=\mathbf{x}_i^\top\mathbf{a}_l+\varepsilon_{il},\qquad
h(Y_{ir})=\mathbf{x}_i^\top\mathbf{a}_r+\varepsilon_{ir},
\end{equation}
where $\mathbf{a}_m,\mathbf{a}_l,\mathbf{a}_r\in\mathbb{R}^p$ are real coefficient vectors and $g,h:[0,\infty)\to[0,\infty)$ are invertible transformations. \citet{Ferraro2010} use $g=h=\log$, which requires strictly positive spreads. Since some recorded spreads in the application are exactly zero, we use $g=h=\log(1+\cdot)$, which is finite at zero. Because the linear predictor of a transformed spread can be negative, the back-transformation is followed by a nonnegative projection, which produces admissible, possibly zero, fitted spreads. Predicted spreads are projected by $\pi(\eta)=\max\{0,\exp(\eta)-1\}$. Componentwise least squares remains the exact minimizer of the quadratic transformed-scale objective before projection, but it need not minimize the original-scale Yang--Ko discrepancy after projection. Model \eqref{eq:model} is the specification of \citet{Ferraro2010} with fuzzy covariates entering through their centers and spreads, along the lines of \citet{FerraroGiordani2012}. Different components of a fuzzy covariate may therefore affect the center and the imprecision of the response in different ways. A crisp covariate enters with its value only, because its spreads are zero. In what follows, $\Phi(\tilde{Y}_i)=(Y_{im},g(Y_{il}),h(Y_{ir}))$ denotes the transformed response and $\mathbf{Y}_\Phi$ the $n\times 3$ matrix with rows $\Phi(\tilde{Y}_i)$.

\begin{assumption}
\label{ass:model}
\begin{itemize}
\item[(A1)] The units $(\mathbf{x}_i,\tilde{Y}_i)$ are independent and identically distributed.
\item[(A2)] $E(\varepsilon_{ik}\mid\mathbf{x}_i)=0$ for $k\in\{m,l,r\}$, with finite second moments.
\item[(A3)] $\mathbf{X}$ has full column rank $p$.
\end{itemize}
\end{assumption}

No distributional assumption is imposed on the residual discrepancies: normality and homoscedasticity are not required, and the three residuals of a unit may be correlated. Assumption (A1) refers to the sampling of units. It is the standard framework of fuzzy random variables \citep{Puri1986,Nather2006,Ferraro2010}, in which the fuzzy value observed on each unit is a data object and not a realization of a latent crisp quantity. Assumption (A1) is what justifies the bootstrap inference of Section~\ref{subsec:inference}. Assumption (A3) ensures identifiability. It fails, for instance, when a covariate is an exact linear function of another one.

\subsection{Estimation}
\label{subsec:estimation}

Let $\mathbf{A}=[\mathbf{a}_m,\mathbf{a}_l,\mathbf{a}_r]$ be the $p\times 3$ coefficient matrix. The estimator minimizes the sum of Yang--Ko discrepancies between the observed and fitted transformed responses,
\begin{equation}
\label{eq:criterion}
S(\mathbf{A})=\sum_{i=1}^n \D\bigl(\Phi(\tilde{Y}_i),\,\mathbf{x}_i^\top\mathbf{A}\bigr),
\end{equation}
where the expression \eqref{eq:yk} is applied to the triples formed by the center and the transformed left and right spreads. Formula \eqref{eq:yk} is a quadratic form in the differences of the three coordinates and is defined for any pair of real triples. The minimization in \eqref{eq:criterion} is over all $\mathbf{A}\in\mathbb{R}^{p\times 3}$, without constraints: a fitted transformed spread $\mathbf{x}_i^\top\mathbf{a}_l$ may be negative and is then not the image under $g$ of an admissible spread. Criterion \eqref{eq:criterion} is therefore the unconstrained quadratic extension of the Yang--Ko distance to $\mathbb{R}^3$, and it coincides with a distance between triangular fuzzy numbers only when all fitted transformed spreads are nonnegative.

\begin{proposition}
\label{prop:ls_yk}
Under (A3), for any $\lambda,\rho>0$ the criterion \eqref{eq:criterion} has the unique unconstrained minimizer
\[
\widehat{\mathbf{A}}=(\mathbf{X}^\top\mathbf{X})^{-1}\mathbf{X}^\top\mathbf{Y}_\Phi,
\]
that is, ordinary least squares applied separately to $Y_m$, $g(Y_l)$ and $h(Y_r)$.
\end{proposition}

\begin{proof}
Let $\mathbf{e}_i=\Phi(\tilde{Y}_i)^\top-\mathbf{A}^\top\mathbf{x}_i=(e_{im},e_{il},e_{ir})^\top$. Expanding \eqref{eq:yk},
\[
\D=e_{im}^2+(e_{im}-\lambda e_{il})^2+(e_{im}+\rho e_{ir})^2=\mathbf{e}_i^\top\mathbf{W}\mathbf{e}_i,
\qquad
\mathbf{W}=\begin{pmatrix}3&-\lambda&\rho\\-\lambda&\lambda^2&0\\\rho&0&\rho^2\end{pmatrix}.
\]
The leading principal minors of $\mathbf{W}$ are $3$, $2\lambda^2$ and $\det\mathbf{W}=\lambda^2\rho^2$, all positive, so $\mathbf{W}$ is positive definite. With $\mathbf{E}=\mathbf{Y}_\Phi-\mathbf{X}\mathbf{A}$, $S(\mathbf{A})=\operatorname{tr}(\mathbf{E}\mathbf{W}\mathbf{E}^\top)$, a quadratic function of $\mathbf{A}$ with Hessian $2(\mathbf{W}\otimes\mathbf{X}^\top\mathbf{X})$, which is positive definite under (A3) of Assumption~\ref{ass:model}. The unique stationary point solves $\mathbf{X}^\top(\mathbf{Y}_\Phi-\mathbf{X}\mathbf{A})\mathbf{W}=\mathbf{0}$, and since $\mathbf{W}$ is non-singular this is equivalent to $\mathbf{X}^\top(\mathbf{Y}_\Phi-\mathbf{X}\mathbf{A})=\mathbf{0}$.
\end{proof}

Proposition~\ref{prop:ls_yk} is the identical-regressors case of seemingly unrelated regressions \citep{Zellner1962}, and the proof uses only the positive definiteness of $\mathbf{W}$. The same estimator therefore minimizes $\sum_i\mathbf{e}_i^\top\mathbf{M}\mathbf{e}_i$ for every positive definite weight matrix $\mathbf{M}$, of which the Yang--Ko distance is one case. Two consequences follow. The estimates, and all effects derived from them in Sections~\ref{sec:mediation} and~\ref{sec:moderation}, do not depend on $\lambda$ and $\rho$, which affect only the goodness-of-fit and validation measures of Section~\ref{subsec:gof}. The common design is essential: if the three components had different regressors, minimizing \eqref{eq:criterion} would require a weighted joint fit that depends on $\mathbf{M}$. The minimized criterion is defined on the transformed scale. The back-transformed projected fitted spreads $\widehat{Y}_{il}=\pi(\mathbf{x}_i^\top\widehat{\mathbf{a}}_l)$ and $\widehat{Y}_{ir}=\pi(\mathbf{x}_i^\top\widehat{\mathbf{a}}_r)$ are nonnegative by projection but are not necessarily the minimizers of the distance on the original scale. With $g=\log(1+\cdot)$ the regression is specified for the transformed spreads and the fitted values are projected onto the nonnegative half-line. The fitted fuzzy response is
\begin{equation}
\label{eq:fitted}
\widehat{\tilde{Y}}_i=\bigl(\mathbf{x}_i^\top\widehat{\mathbf{a}}_m,\; \pi(\mathbf{x}_i^\top\widehat{\mathbf{a}}_l),\; \pi(\mathbf{x}_i^\top\widehat{\mathbf{a}}_r)\bigr).
\end{equation}
Under Assumption~\ref{ass:model}, $\widehat{\mathbf{A}}$ is consistent by standard least squares arguments; see \citet{Ferraro2010} for crisp covariates and \citet{Kratschmer2006a,Kratschmer2006b} for related results.

\begin{remark}
\label{rem:scale}
If all spreads of a fuzzy covariate are multiplied by $c>0$, the fitted values in \eqref{eq:fitted} are unchanged and the corresponding coefficients are divided by $c$. This is the usual invariance of least squares fits to the rescaling of a regressor. The fit, the effects of Sections~\ref{sec:mediation} and~\ref{sec:moderation} and their precision depend on the variability of the covariate spreads across units and on the size of the residual discrepancies, not on the overall scale of the spreads. For this reason the simulation study of Section~\ref{sec:simulation} varies the residual scale.
\end{remark}

\subsection{Interpretation on the original scale}
\label{subsec:interpretation}

The interpretation of coefficients in the center equation follows the usual linear logic. For instance, the coefficient $a_{m,X_{tm}}$ of the center component $X_{tm}$ quantifies the change in the center of the response associated with a one unit increase in $X_{tm}$, conditional on the other covariate components. For the spread equations, interpretation is naturally expressed on the transformed scale. With $g=h=\log(1+\cdot)$, a coefficient $b$ represents a change in $\log(1+Y_l)$ for a unit increase of the corresponding regressor: in the absence of projection, one plus the fitted left spread is multiplied by $\exp(b)$. An analogous interpretation applies to the right spread. Negative coefficients do not produce negative fitted spreads after projection: they indicate a reduction of the imprecision of the response, with multiplicative effects smaller than one on the original scale.

\subsection{Goodness of fit and validation}
\label{subsec:gof}

Consistently with Proposition~\ref{prop:ls_yk}, goodness of fit is measured on the transformed scale by
\begin{equation}
\label{eq:r2}
R^2_{\Tsc}=1-\frac{\sum_{i=1}^n\D\bigl(\Phi(\tilde{Y}_i),\mathbf{x}_i^\top\widehat{\mathbf{A}}\bigr)}{\sum_{i=1}^n\D\bigl(\Phi(\tilde{Y}_i),\overline{\Phi}\bigr)},
\qquad
\overline{\Phi}=\frac1n\sum_{i=1}^n\Phi(\tilde{Y}_i).
\end{equation}
Since the intercept-only model is nested in every model with an intercept and the estimator minimizes the numerator, $0\le R^2_{\Tsc}\le 1$, and $R^2_{\Tsc}$ cannot decrease when regressors are added. The index is analogous to the determination coefficient of \citet{Ferraro2011}. It is not a measure of predictive ability: since it cannot decrease when regressors are added, an in-sample increase is not by itself evidence of a better model. We therefore also report the adjusted version $1-(1-R^2_{\Tsc})(n-1)/(n-p)$ and leave-one-out cross-validated distances. Since the three equations share the hat matrix $\mathbf{H}=\mathbf{X}(\mathbf{X}^\top\mathbf{X})^{-1}\mathbf{X}^\top$, the leave-one-out prediction of $\Phi(\tilde{Y}_i)$ has the closed form $\Phi(\tilde{Y}_i)-\widehat{\mathbf{e}}_i/(1-h_{ii})$, with $\widehat{\mathbf{e}}_i$ the vector of residuals and $h_{ii}$ the leverage. We define
\begin{equation}
\label{eq:cv}
\mathrm{CV}_{\Tsc}=\frac1n\sum_{i=1}^n\D\bigl(\Phi(\tilde{Y}_i),\widehat{\Phi}_{(-i)}\bigr),
\qquad
\mathrm{CV}_{\Osc}=\frac1n\sum_{i=1}^n\D\bigl(\tilde{Y}_i,\widehat{\tilde{Y}}_{(-i)}\bigr),
\end{equation}
where $\widehat{\tilde{Y}}_{(-i)}$ is the back-transformed leave-one-out prediction. $\mathrm{CV}_{\Osc}$ is the out-of-sample discrepancy between observed and predicted triangles on the original scale.

\subsection{Inference}
\label{subsec:inference}

Uncertainty about coefficients and derived effects is quantified by the nonparametric case bootstrap \citep{EfronTibshirani1993}. Units are resampled with replacement, which keeps the three components of each fuzzy observation together and preserves their dependence. We report 95\% percentile intervals. To test whether a block of regressors, such as all interaction components in Section~\ref{sec:moderation}, contributes to the model, we use the statistic
\begin{equation}
\label{eq:fstat}
F_D=\frac{(S_0-S_1)/(p_1-p_0)}{S_1/(n-p_1)},
\end{equation}
where $S_0$ and $S_1$ are the minimized criteria \eqref{eq:criterion} of the reduced and full models, with $p_0<p_1$ parameters per equation. Its null distribution is approximated by the permutation scheme of \citet{FreedmanLane1983}. The rows of the residual matrix of the reduced model are permuted, the permuted residuals are added to the reduced-model fitted values, and $F_D$ is recomputed. Permuting entire rows keeps the cross-component dependence of each unit. The test requires an assumption stronger than Assumption~\ref{ass:model}:
\begin{assumption}
\label{ass:exch}
\begin{itemize}
\item[(A4)] Under the null hypothesis, the residual vectors $(\varepsilon_{im},\varepsilon_{il},\varepsilon_{ir})$ of the reduced model are exchangeable across units.
\end{itemize}
\end{assumption}
Assumption (A4) excludes, in particular, heteroscedasticity that depends on the covariates. Even under (A4) the test is not exact in finite samples, because the permuted quantities are estimated residuals; its size in a design that satisfies (A4) is assessed in Section~\ref{subsec:sim_results}. Under heteroscedasticity its level is not guaranteed. The coefficient intervals in the tables are marginal and not adjusted for multiplicity. The effects of main interest are the indirect effects and the interaction blocks; the remaining intervals are read descriptively.

\section{Mediation effects in fuzzy regression}
\label{sec:mediation}

\subsection{General structure and componentwise effects}
\label{subsec:mediation_general}

Let $\tilde{Y}$ denote a fuzzy response, let $X$ denote an explanatory variable and let $Z$ denote a mediator. Depending on the application, $X$ and $Z$ may be fuzzy numbers, crisp numerical variables or crisp qualitative variables, and $\mathbf{C}$ denotes additional crisp covariates. In the classical formulation \citep{BaronKenny1986}, mediation is described by a system of regression relations: a total relation of $\tilde{Y}$ on $X$, a mediator relation of $Z$ on $X$, and an outcome relation of $\tilde{Y}$ on $X$ and $Z$. The outcome relation decomposes the association between $X$ and $\tilde{Y}$ into a direct component and an indirect component through $Z$.

In the fuzzy setting, each relation is a model of the form \eqref{eq:model}. When $\tilde{Z}$ is fuzzy, the mediator system is
\begin{equation}
\label{eq:med_system}
\Phi(\tilde{Z}_i)=\mathbf{w}_i^\top\boldsymbol{\Gamma}+\mathbf{v}_i,
\qquad
\Phi(\tilde{Y}_i)=\mathbf{w}_i^\top\mathbf{B}_W+\mathbf{z}_i^\top\mathbf{B}_Z+\mathbf{u}_i,
\end{equation}
where $\mathbf{w}_i$ collects the intercept, the components of $X$ and $\mathbf{C}$, $\mathbf{z}_i^\top=(Z_{im},Z_{il},Z_{ir})$ collects the components of the mediator on the original scale, and $\boldsymbol{\Gamma}$, $\mathbf{B}_W$, $\mathbf{B}_Z$ are real coefficient matrices with three columns (center, transformed left spread, transformed right spread). The indirect association between $X$ and $\tilde{Y}$ can pass through three channels of the mediator, its center and its two spreads, and each of them can affect the center and the imprecision of the response. A single product of two path coefficients cannot summarize these paths. Moreover, when the spreads are transformed, the change in the mediator induced by $X$ is not linear in $\boldsymbol{\Gamma}$. We therefore define effects as contrasts between fitted values.

\begin{definition}
\label{def:effects}
Let $\delta\neq 0$ be a change in one component of $X$, here the center $X_m$, with the other components held fixed. For unit $i$, let $\mathbf{z}_i$ be the observed mediator and $\mathbf{z}_i(\delta)$ the mediator obtained by moving $X_{im}$ to $X_{im}+\delta$ in the mediator equations while keeping the observed residual:
\begin{align*}
Z_{im}(\delta)&=Z_{im}+\delta\,\Gamma_{X_m,m},\\
Z_{il}(\delta)&=\pi\bigl(g(Z_{il})+\delta\,\Gamma_{X_m,l}\bigr),\\
Z_{ir}(\delta)&=\pi\bigl(h(Z_{ir})+\delta\,\Gamma_{X_m,r}\bigr).
\end{align*}
Let $f_k(\mathbf{w},\mathbf{z})$, $k\in\{m,l,r\}$, be the $k$-th component of the fitted response \eqref{eq:fitted} on the original scale, and let $\mathbf{w}_i(\delta)$ be $\mathbf{w}_i$ with $X_{im}$ replaced by $X_{im}+\delta$. The model-based direct effect (NDE), indirect effect (NIE) and total effect (TE) on component $k$ are
\begin{align*}
\mathrm{NDE}_k(\delta)&=\frac1n\sum_{i=1}^n\bigl[f_k(\mathbf{w}_i(\delta),\mathbf{z}_i)-f_k(\mathbf{w}_i,\mathbf{z}_i)\bigr],\\
\mathrm{NIE}_k(\delta)&=\frac1n\sum_{i=1}^n\bigl[f_k(\mathbf{w}_i(\delta),\mathbf{z}_i(\delta))-f_k(\mathbf{w}_i(\delta),\mathbf{z}_i)\bigr],\\
\mathrm{TE}_k(\delta)&=\mathrm{NDE}_k(\delta)+\mathrm{NIE}_k(\delta).
\end{align*}
\end{definition}

The contrasts are meaningful when the shifted covariate $X_{im}+\delta$ and the counterfactual mediator $\mathbf{z}_i(\delta)$ remain in the domains of the corresponding scales. For bounded scales this restricts $\delta$, and the share of units for which the shifted values leave the domain should be reported together with the effects, as is done in Section~\ref{sec:application}.

The effects are plug-in sample averages, estimated by replacing $\boldsymbol{\Gamma}$, $\mathbf{B}_W$ and $\mathbf{B}_Z$ with their estimators. The decomposition $\mathrm{TE}_k=\mathrm{NDE}_k+\mathrm{NIE}_k$ holds by construction for every component. Effects on the spreads are expressed on the original scale of the spreads, which makes direct and indirect contributions to the imprecision of the response directly comparable. The next result relates Definition~\ref{def:effects} to the product-of-coefficients rule.

\begin{proposition}
\label{prop:decomp}
Consider the center component of the response.
\begin{enumerate}
\item[(i)] $\mathrm{NDE}_m(\delta)=\delta\,B_{W;X_m,m}$ and $\mathrm{NIE}_m(\delta)=\sum_{j\in\{m,l,r\}}B_{Z;j,m}\,\overline{\Delta Z}_j(\delta)$, where $\overline{\Delta Z}_j(\delta)=n^{-1}\sum_i[Z_{ij}(\delta)-Z_{ij}]$. In particular, $\overline{\Delta Z}_m(\delta)=\delta\,\Gamma_{X_m,m}$ and, for $g=\log(1+\cdot)$, $\overline{\Delta Z}_l(\delta)=n^{-1}\sum_i[\pi\{\log(1+Z_{il})+\delta\Gamma_{X_m,l}\}-Z_{il}]$, which reduces to $n^{-1}\sum_i(1+Z_{il})\{\exp(\delta\Gamma_{X_m,l})-1\}$ when the projection is not active. The right spread is analogous.
\item[(ii)] If the mediator spreads are not transformed in the mediator equations, then $\mathrm{NIE}_m(\delta)=\delta\sum_{j}\Gamma_{X_m,j}B_{Z;j,m}$. If, in addition, the total relation $Y_{im}=\mathbf{w}_i^\top\boldsymbol{\alpha}+\epsilon_i$ uses the same $\mathbf{w}_i$ and all equations are fitted by least squares, then $\widehat{\mathrm{TE}}_m(\delta)=\delta\,\widehat{\alpha}_{X_m}$ exactly.
\end{enumerate}
\end{proposition}

\begin{proof}
Part (i) follows from the linearity of $f_m$ in $\mathbf{w}$ and $\mathbf{z}$. For (ii), the first claim follows from (i) with $Z_{ij}(\delta)-Z_{ij}=\delta\Gamma_{X_m,j}$. For the second claim, let $\mathbf{W}_0$ and $\mathbf{Z}_0$ be the matrices with rows $\mathbf{w}_i^\top$ and $\mathbf{z}_i^\top$. Then $Y_m=\mathbf{W}_0\widehat{\mathbf{b}}_W+\mathbf{Z}_0\widehat{\mathbf{b}}_Z+\widehat{\mathbf{u}}$ with $\widehat{\mathbf{u}}$ orthogonal to the columns of $\mathbf{W}_0$. Hence $\widehat{\boldsymbol{\alpha}}=(\mathbf{W}_0^\top\mathbf{W}_0)^{-1}\mathbf{W}_0^\top Y_m=\widehat{\mathbf{b}}_W+\widehat{\boldsymbol{\Gamma}}\widehat{\mathbf{b}}_Z$, where $\widehat{\boldsymbol{\Gamma}}=(\mathbf{W}_0^\top\mathbf{W}_0)^{-1}\mathbf{W}_0^\top\mathbf{Z}_0$, and the row of $X_m$ gives the result.
\end{proof}

Proposition~\ref{prop:decomp} gives a sufficient condition for the classical identity ``total $=$ direct $+$ product of paths'' on the center: the mediator spreads are modeled on their original scale. The identity also holds, trivially, when the paths through the mediator spreads are inactive, that is, when $B_{Z;l,m}=B_{Z;r,m}=0$ or $\Gamma_{X_m,l}=\Gamma_{X_m,r}=0$; then $\mathrm{NIE}_m(\delta)=\delta\,\Gamma_{X_m,m}B_{Z;m,m}$. When these paths are active and the mediator spreads are transformed, the indirect effect through the spreads depends on the observed spreads, and the plug-in total effect differs from the coefficient of the total relation. For the spread components of the response, $f_l$ and $f_r$ are nonlinear, and Definition~\ref{def:effects} is the natural generalization.

Mediation is called total for component $k$ when $\mathrm{NIE}_k\neq0$ and $\mathrm{NDE}_k=0$, and partial when both are different from zero. In finite samples these statements are assessed with bootstrap intervals. Failing to exclude zero for $\mathrm{NDE}_k$ is not evidence of total mediation. The acronyms NDE and NIE follow the natural direct and indirect effects of causal mediation analysis \citep{VanderWeele2015}, but Definition~\ref{def:effects} is a descriptive decomposition of associations. A causal reading of NDE and NIE would require sequential ignorability, namely no unmeasured confounding of the $X$--$Z$, $X$--$Y$ and $Z$--$Y$ relations and no mediator--outcome confounder affected by $X$ \citep{ImaiEtAl2010,VanderWeele2015}. These conditions are not assumed in this paper.

\subsection{Fuzzy predictor and fuzzy mediator}

When both the predictor $X$ and the mediator $Z$ are represented as fuzzy numbers, mediation is modeled through \eqref{eq:med_system} with $\mathbf{w}_i^\top=(1,X_{im},X_{il},X_{ir},\mathbf{C}_i^\top)$. In compact notation, the structure can be written as
\begin{equation}
\label{eq:fuzzy_fuzzy_mediation_system}
\tilde{Z}_i = \gamma_0 + \gamma_1 \tilde{X}_i + \tilde{u}_{Zi},
\qquad
\tilde{Y}_i = \beta_0 + \beta_1 \tilde{X}_i + \beta_2 \tilde{Z}_i + \tilde{u}_{Yi},
\end{equation}
where each product of a coefficient and a fuzzy variable is shorthand for a linear combination of its three components, with coefficients specific to each component equation of the left-hand side. The first relation describes the association between the fuzzy predictor and the fuzzy mediator, while the second decomposes the association between $\tilde{X}$ and $\tilde{Y}$ into a direct component and an indirect component associated with $\tilde{Z}$. The conceptual structure is illustrated in Figure~\ref{fig:fuzzy_mediation_total_partial}. The upper panel represents total mediation, with no direct path, and the lower panel partial mediation. Effects are computed componentwise with Definition~\ref{def:effects}.

\begin{figure}[htbp]
\centering
\includegraphics[width=0.62\textwidth]{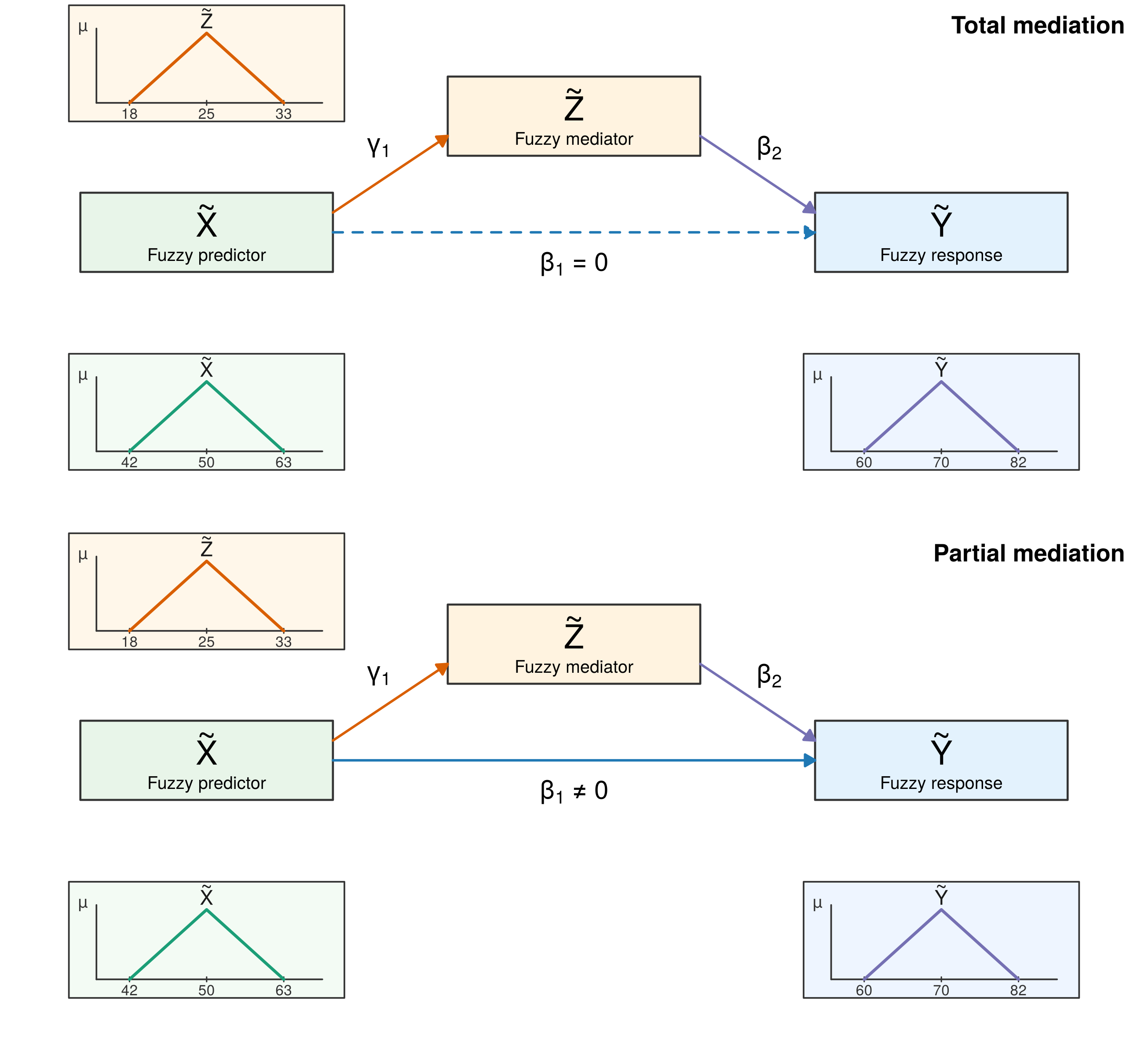}
\caption{Total (upper panel) and partial (lower panel) mediation with fuzzy predictor, fuzzy mediator and fuzzy response, all represented by triangular fuzzy numbers. Labels $\gamma_1$, $\beta_1$, $\beta_2$ stand for the blocks of componentwise coefficients in \eqref{eq:fuzzy_fuzzy_mediation_system}; $\beta_1=0$ means that all coefficients of the components of $\tilde{X}$ in the outcome equations are zero, so that $\mathrm{NDE}_k=0$ for every $k$.}
\label{fig:fuzzy_mediation_total_partial}
\end{figure}

\subsection{Crisp predictor and fuzzy mediator}
\label{subsec:mediation_crispX_fuzzyZ}

In several applications the explanatory variable is observed without imprecision and is therefore treated as crisp, while the mediator represents an intermediate construct that is more appropriately elicited in linguistic or approximate form and hence modeled as a fuzzy number. The crisp predictor may be quantitative, for instance age or a test score, or qualitative, for instance group membership, in which case it is introduced through dummy variables $G_{ih}$, $h=1,\ldots,H-1$. System \eqref{eq:med_system} applies with $\mathbf{w}_i$ containing $X_i$ or the dummies. For a quantitative predictor, Definition~\ref{def:effects} is applied with $\delta$ equal to a change in $X_i$. For a qualitative predictor, the contrast is between category $h$ and the reference category, that is, $\mathbf{w}_i(\delta)$ sets $G_{ih}=1$ and $\mathbf{w}_i$ sets all dummies to zero, which gives category-specific direct and indirect effects. Figure~\ref{fig:fuzzy_mediation_crispX_total_partial} in \ref{app:figures} illustrates total and partial mediation in the quantitative case.

\subsection{Fuzzy predictor and crisp mediator}
\label{subsec:mediation_fuzzyX_crispZ}

When the predictor is modeled as a fuzzy number while the mediator is observed without imprecision, the mediator may be quantitative or qualitative. For a quantitative mediator $Z_i$, the mediator equation is the ordinary regression $Z_i=\mathbf{w}_i^\top\boldsymbol{\gamma}+v_i$, with $\mathbf{w}_i$ containing the three components of $\tilde{X}_i$. The outcome equations contain $Z_i$ as a crisp regressor, and Definition~\ref{def:effects} applies with $Z_i(\delta)=Z_i+\delta\gamma_{X_m}$.

For a qualitative mediator with $K$ categories, coded by dummies $D_{ik}$, $k=1,\ldots,K-1$, the outcome equations are
\begin{equation}
\label{eq:med_fuzzyX_crispZ_qual}
\Phi(\tilde{Y}_i)=\mathbf{w}_i^\top\mathbf{B}_W+\sum_{k=1}^{K-1}D_{ik}\,\boldsymbol{\beta}_{2k}^\top+\mathbf{u}_i.
\end{equation}
A mediator model is still needed, for instance a multinomial logit model for $p_c(\mathbf{w}_i)=P(Z_i=c\mid\mathbf{w}_i)$, $c=1,\ldots,K$. Comparing the direct association of $\tilde{X}$ with and without the dummies (the difference method) does not identify an indirect effect. With fitted probabilities $\widehat{p}_c$, the model-based indirect effect on component $k\in\{m,l,r\}$ is
\[
\mathrm{NIE}_{k}(\delta)=\frac1n\sum_{i=1}^n\sum_{c=1}^{K}\bigl[\widehat{p}_c(\mathbf{w}_i(\delta))-\widehat{p}_c(\mathbf{w}_i)\bigr]f_{k}\bigl(\mathbf{w}_i(\delta),c\bigr),
\]
which averages the fitted response over the shift in the category probabilities induced by the predictor. With a quantitative mediator, mediation reflects a continuous transmission mechanism. With a qualitative mediator, it reflects a change in the composition of groups with different fuzzy response levels. The structure is illustrated in Figure~\ref{fig:fuzzy_mediation_fuzzyX_qualZ_total_partial} in \ref{app:figures}.

\section{Moderation effects in fuzzy regression}
\label{sec:moderation}

In many empirical settings, the effect of an explanatory variable on a response is not homogeneous across the population, but varies systematically as a function of a third variable, the moderator. In classical regression, moderation is modeled through interaction terms \citep{BaronKenny1986,Hayes2018}. Let $\tilde{Y}_i$ be a fuzzy response, $X_i$ an explanatory variable and $M_i$ a moderator. The general fuzzy regression model with moderation is
\begin{equation}
\label{eq:fuzzy_moderation_general}
\tilde{Y}_i=\beta_0+\beta_1 X_i+\beta_2 M_i+\beta_3\bigl(X_i\circ M_i\bigr)+\tilde{u}_i,
\end{equation}
where $\circ$ is an interaction operator whose definition depends on the nature of the variables involved, and, as in \eqref{eq:fuzzy_fuzzy_mediation_system}, products of coefficients and fuzzy variables are shorthand for componentwise linear combinations within model \eqref{eq:model}. Moderation is interpreted as a structural change in the relationship between $X$ and the fuzzy response rather than as a transmission mechanism.

The interaction regressor must be a TFN, so that its center and spreads can enter the design \eqref{eq:design_row}. We use the following operators.
\begin{itemize}
\item \emph{Fuzzy predictor, crisp numerical moderator $M\ge0$}: $M\odot\tilde{X}=(MX_m,MX_l,MX_r)$, exact by \eqref{eq:sum_scalar}. If $M$ takes negative values, for instance after centering, the spreads are swapped as in \eqref{eq:sum_scalar}.
\item \emph{Fuzzy predictor, crisp categorical moderator}: $D_k\odot\tilde{X}$ for each dummy $D_k$, again exact.
\item \emph{Fuzzy predictor, fuzzy moderator}: the Dubois--Prade approximation \eqref{eq:dp_product}, $\tilde{X}\otimes\tilde{M}=(X_mM_m,\,X_mM_l+M_mX_l,\,X_mM_r+M_mX_r)$, for positive TFNs.
\end{itemize}
In all cases, $\beta_3(X\circ M)$ contributes three regressors (the components of the interaction TFN) to each of the three equations of \eqref{eq:model}. Moderation is absent when all of these coefficients are zero, which is tested with the Freedman--Lane test of Section~\ref{subsec:inference}.

\paragraph{Conditional effects} For a crisp moderator, the change in the linear predictor of component $k$ associated with an increase $\delta$ in $X_m$ is $\delta(a_{k,X_m}+a_{k,(X\circ M)_m}M)$. For the center this is the usual simple slope. For spread coordinates it is a difference on the log-one-plus scale; without active projection, the original-scale contrast is $(1+f_k(\cdot))[\exp\{\delta(a_{k,X_m}+a_{k,(X\circ M)_m}M)\}-1]$. For a fuzzy moderator, the center of $X$ enters all three components of $\tilde{X}\otimes\tilde{M}$, and the change becomes
\begin{equation}
\label{eq:cond_fuzzy}
\delta\bigl(a_{k,X_m}+a_{k,(X\circ M)_m}M_m+a_{k,(X\circ M)_l}M_l+a_{k,(X\circ M)_r}M_r\bigr),
\end{equation}
so that the effect of the predictor depends on both the center and the imprecision of the moderator. In the application, conditional effects for the CR and CI blocks are reported on the original scale at the quartiles of the moderator, with the other covariates at their medians and case-bootstrap intervals.

\subsection{Moderation between fuzzy predictors}
\label{subsec:moderation_fuzzy_fuzzy}

When both the explanatory variable and the moderator are represented by fuzzy numbers, moderation arises from the interaction between two imprecise quantities. From a strictly mathematical standpoint, the exact product of two TFNs, defined through the extension principle or equivalently via $\alpha$-cuts, does not preserve the triangular shape (Section~\ref{subsec:operations}). This does not pose a problem in the distance based framework adopted here. The interaction term is not intended to represent the exact fuzzy product, but a fuzzy regressor constructed deterministically from $\tilde{X}_i$ and $\tilde{M}_i$ so as to remain in the space of TFNs. The approximation \eqref{eq:dp_product} preserves the core exactly, and its error on the support, $X_lM_l$ on the left and $X_rM_r$ on the right, is small when the spreads are small relative to the centers. Such a construction ensures that all regressors, including the interaction term, belong to the same space and can therefore be handled coherently within the proposed estimation procedure. Figure~\ref{fig:fuzzy_moderation_fuzzy_fuzzy} illustrates the absence and the presence of moderation.

\begin{figure}[htbp]
\centering
\includegraphics[width=0.6\textwidth]{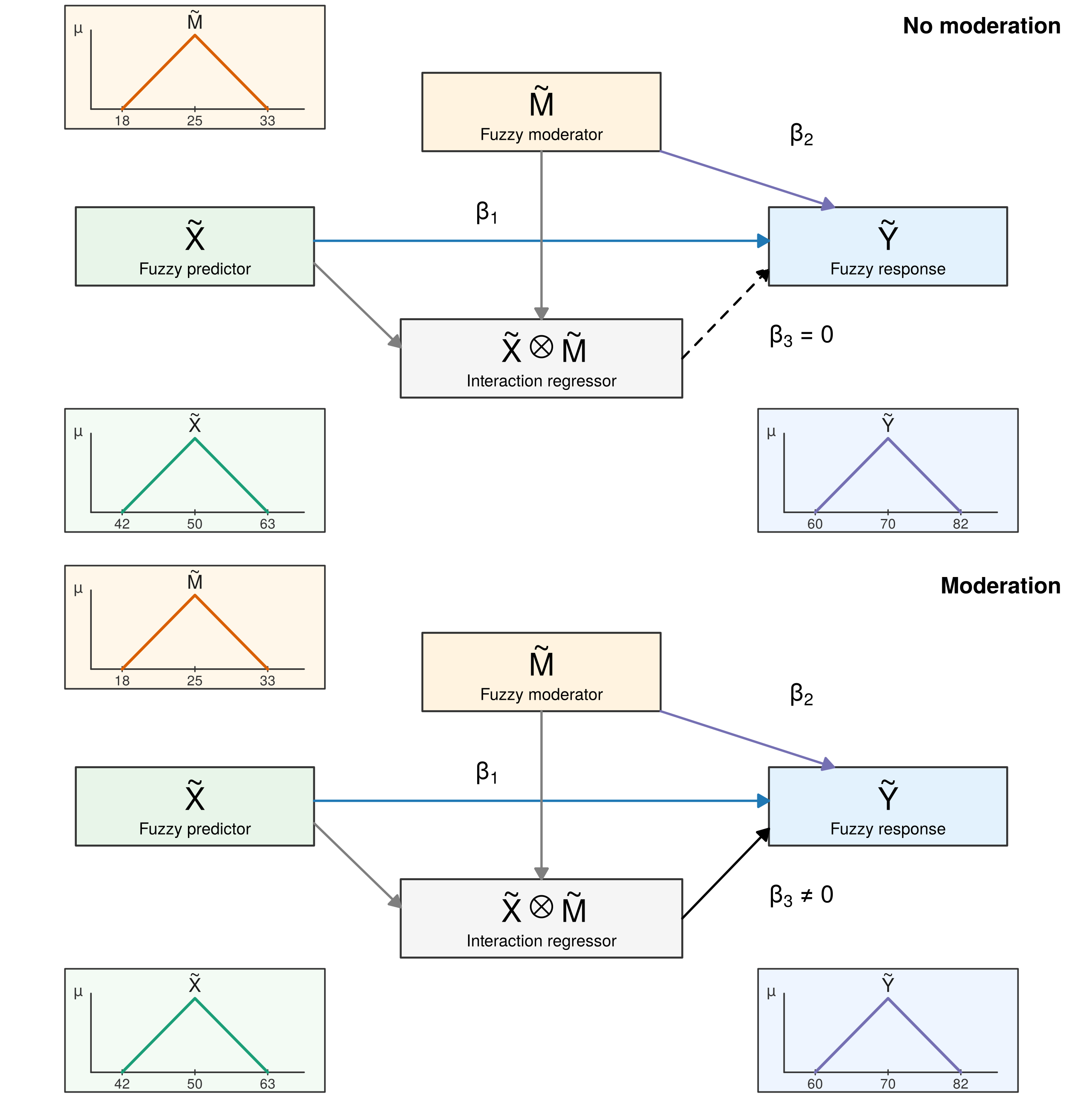}
\caption{Moderation with fuzzy predictor and fuzzy moderator. The interaction term is the additional fuzzy regressor $\tilde{X}\otimes\tilde{M}$ of \eqref{eq:dp_product}. Upper panel, absence of moderation, $\beta_{3}=0$. Lower panel, presence of moderation, $\beta_{3}\neq 0$.}
\label{fig:fuzzy_moderation_fuzzy_fuzzy}
\end{figure}

\subsection{Moderation between a fuzzy predictor and a crisp numerical moderator}
\label{subsec:moderation_fuzzy_crisp_num}

In many empirical applications, the effect of an imprecisely observed construct depends on the level of a precisely measured contextual variable. Let $\tilde{X}_i$ denote a fuzzy predictor and $M_i$ a crisp numerical moderator. The moderated model is
\begin{equation}
\label{eq:fuzzy_moderation_fuzzy_crisp_num}
\tilde{Y}_i=\beta_0+\beta_1\tilde{X}_i+\beta_2M_i+\beta_3\bigl(M_i\odot\tilde{X}_i\bigr)+\tilde{u}_i
=\beta_0+\beta_2M_i+\bigl(\beta_1+\beta_3M_i\bigr)\tilde{X}_i+\tilde{u}_i,
\end{equation}
so that the contribution of the fuzzy predictor varies continuously with $M_i$, in direct analogy with classical interaction models. Although the interaction term is algebraically symmetric in the two variables, the substantive interpretation depends on which variable is assigned the role of moderator. To preserve conceptual clarity, attention is restricted to the formulation in which the fuzzy variable plays the role of predictor and the crisp variable acts as moderator. Figure~\ref{fig:fuzzy_moderation_fuzzy_crisp_num} in \ref{app:figures} illustrates the two situations.

\subsection{Moderation between a fuzzy predictor and a crisp categorical moderator}
\label{subsec:moderation_fuzzy_crisp_qual}

Let $M_i\in\{1,\ldots,K\}$ be a crisp categorical moderator encoded through dummies $D_{ik}=\mathbb{I}(M_i=k)$, $k=1,\ldots,K-1$. The moderated model is
\begin{equation}
\label{eq:fuzzy_moderation_fuzzy_crisp_qual}
\tilde{Y}_i=\beta_0+\beta_1\tilde{X}_i+\sum_{k=1}^{K-1}\beta_{2k}D_{ik}+\sum_{k=1}^{K-1}\beta_{3k}\bigl(D_{ik}\odot\tilde{X}_i\bigr)+\tilde{u}_i.
\end{equation}
The coefficients $\beta_{2k}$ represent group specific shifts of the fuzzy response with respect to the reference category, and $\beta_{3k}$ represent deviations of the effect of the fuzzy predictor in category $k$, so that the reference category has slope $\beta_1$ and category $k$ has slope $\beta_1+\beta_{3k}$, componentwise. For a binary moderator the model reduces to two group specific fuzzy regressions. Figure~\ref{fig:fuzzy_moderation_qualitativeM} in \ref{app:figures} illustrates the case.

\section{Simulation study}\label{sec:simulation}

\subsection{Objectives and design}\label{subsec:sim_design}
The simulation addresses three questions that the application cannot answer, because there the true effects are unknown:
\begin{itemize}
\item[(Q1)] whether the percentile case-bootstrap intervals for the plug-in effects of Definition~\ref{def:effects} attain their nominal coverage, and their tests their nominal size, at realistic sample sizes;
\item[(Q2)] how much is lost by estimating the indirect effect on the center with crisp analyses that ignore some or all of the spreads, and which omitted components are responsible;
\item[(Q3)] whether the Freedman--Lane test of Section~\ref{subsec:inference} controls its size and has power against interactions of realistic magnitude.
\end{itemize}
To make the scenarios realistic, the data-generating processes are calibrated on the data of Section~\ref{sec:application}, whose variables are academic performance (AP), prior knowledge (PK), the course indicator (CI), motivation (MOT), course attendance (CR) and sex (SX), described in Section~\ref{subsec:data}. Two experiments were conducted, one for mediation (Q1, Q2) and one for moderation (Q3). All settings are those of the R code available from the corresponding author, which stores the random seed and the numbers of replications with the results.

\paragraph{Mediation experiment} The data-generating process was calibrated on the data of Section~\ref{sec:application}. For each simulated unit, a row of the data was drawn with replacement and its PK triple $(\mathrm{PK}_m,\mathrm{PK}_l,\mathrm{PK}_r)$ and its CR value were retained jointly, so that the dependence between the center and the spreads of PK and CR is preserved and the support of PK remains in $[0,100]$. The mediator CI was generated from the mediator equations, $\Phi(\widetilde{\mathrm{CI}}_i)=\mathbf{w}_i^\top\boldsymbol{\Gamma}+k\,\mathbf{v}_i$, and the response AP from the outcome equations, $\Phi(\widetilde{\mathrm{AP}}_i)=\mathbf{w}_i^\top\mathbf{B}_W+\mathbf{z}_i^\top\mathbf{B}_Z+k\,\mathbf{u}_i$, where the coefficient matrices are the least squares estimates on the data, $\mathbf{v}_i$ and $\mathbf{u}_i$ are trivariate Gaussian with the empirical residual covariance matrices of the two systems, and spreads are obtained by the projection $\pi$. The noise multiplier was $k\in\{1,2\}$ and the sample size $n\in\{50,100,250,500\}$. In the \emph{empirical} scenario all coefficients were kept; in the \emph{null} scenario the coefficients of the three PK components in the mediator equations were set to zero, so that all indirect effects are zero; in the \emph{no-spread-path} scenario the coefficients of the spreads of CI in the outcome equations were set to zero, so that the whole indirect effect passes through the center of CI. The last scenario is a control for (Q2): an estimator that ignores the spreads of the mediator is consistent there, and any remaining bias must come from other omissions. The true effects of Definition~\ref{def:effects}, for $\delta=10$, were approximated by evaluating the plug-in contrasts with the true parameters on one simulated sample of $2\times10^5$ units. Each design cell used $R=200$ replications and $B=199$ bootstrap resamples per replication. Because the Gaussian residuals are unbounded, generated centers and supports can fall outside the admissible ranges $[0,10]$ for CI and $[0,30]$ for AP. The proportion of such units was recorded, and the generated values were not truncated.

\paragraph{Moderation experiment} This is a plasmode design. In each replication, $n$ complete covariate vectors (PK, CI, CR, SX, MOT) were resampled with replacement from the data, so that the moderators are never generated. For each interaction block, a synthetic response AP$^*$ was constructed as the fitted value of the multivariate model of Section~\ref{subsec:multivariate}, plus a resampled row of its residual matrix, plus the planted interaction of that block only:
\begin{gather*}
\Phi(\widetilde{\mathrm{AP}}^{*}_i)=\mathbf{x}_i^\top\widehat{\mathbf{A}}+\bigl(\mathrm{PK}_{im}-\overline{\mathrm{PK}}_m\bigr)\,Q_i\,\boldsymbol{\theta}_Q^\top+\widehat{\mathbf{e}}_{j(i)}^\top,\\
Q_i\in\bigl\{\mathrm{SX}_i,\;\mathrm{MOT}_{im}-\overline{\mathrm{MOT}}_m\bigr\},
\end{gather*}
with $\boldsymbol{\theta}_{\mathrm{SX}}=(0.10,\,0,\,-0.010)$ and $\boldsymbol{\theta}_{\mathrm{MOT}}=(0.03,\,0,\,0)$ for the center, transformed left spread and transformed right spread. Generating one response per block makes each fitted single-block model correctly specified. The planted center values are of the same order as the empirical estimates of Section~\ref{subsec:mod_app}, so the experiment evaluates the test at effect sizes comparable to those observed. In the null scenario both vectors are zero. The fitted models are those of Section~\ref{subsec:mod_app}; centering does not change the interaction coefficients, which are therefore directly comparable with $\boldsymbol{\theta}$. Sample sizes were as above, with $R=200$ replications and $P=199$ permutations per test.

\subsection{Competing methods}\label{subsec:sim_competitors}
For the indirect effect on the center, the proposed estimator was compared with two crisp product-of-coefficients estimators of the form $\widehat{a}\,\widehat{b}\,\delta$, where $\widehat{a}$ is the coefficient of $\mathrm{PK}_m$ in a regression of $\mathrm{CI}_m$ and $\widehat{b}$ the coefficient of $\mathrm{CI}_m$ in a regression of $\mathrm{AP}_m$. In the centers-only estimator (CO) both regressions contain the centers of PK and CI and CR. In the intermediate estimator (IM) both regressions contain all three components of PK and CR, while the mediator still enters only through its center. The difference between CO and IM isolates the omission of the spreads of PK; the difference between IM and the fuzzy estimator isolates the omission of the paths through the spreads of CI. For prediction, the full fuzzy model was compared with a centers-only model in which the center of AP is predicted by least squares on the centers and the spreads are predicted by their training means. Both predictors were evaluated on an independent test sample of 1000 units generated from the same scenario.

\subsection{Evaluation criteria}\label{subsec:sim_criteria}
For $\mathrm{NDE}_m$, $\mathrm{NIE}_m$, $\mathrm{NIE}_l$ and $\mathrm{NIE}_r$ we report bias, root mean squared error (RMSE), coverage of the 95\% percentile bootstrap interval and the rate of intervals excluding zero, which estimates power in the empirical scenario and size in the null scenario. Prediction is assessed by the mean original-scale Yang--Ko distance on the test samples. For moderation we report bias and RMSE of the coefficients of the interaction centers $(\mathrm{PK}\circ\mathrm{SX})_m$ and $(\mathrm{PK}\circ\mathrm{MOT})_m$ and the rejection rate of the Freedman--Lane test at level 0.05. Monte Carlo standard errors (MCSE) were computed for all bias, coverage and rejection estimates; with $R=200$, the MCSE of a coverage estimate near 0.95 is about 0.015 and that of a rejection rate near 0.05 is about 0.015.

\subsection{Simulation results}\label{subsec:sim_results}

\begin{table}[htbp]
\centering
\caption{Monte Carlo study, mediation scenario calibrated on the empirical estimates ($R=200$ replications, $B=199$ bootstrap resamples; largest Monte Carlo standard error of a coverage estimate 0.021). Cov.: coverage of the 95\% percentile bootstrap interval; Rej.: rate of intervals excluding zero (power). Effects refer to an increase of 10 points in the center of PK. True values: NDE$_m$ 0.736 ($k=1$), 0.736 ($k=2$); NIE$_m$ 0.809 ($k=1$), 0.815 ($k=2$); NIE$_l$ 0.090 ($k=1$), 0.091 ($k=2$); NIE$_r$ 0.031 ($k=1$), 0.032 ($k=2$).}
\label{tab:mc_mediation}
\small
\resizebox{\textwidth}{!}{%
\begin{tabular}{rr|cc|cccc|c|cc}
\toprule
 & & \multicolumn{2}{c|}{NDE$_m$} & \multicolumn{4}{c|}{NIE$_m$} & NIE$_l$ & \multicolumn{2}{c}{NIE$_r$} \\
$n$ & $k$ & Bias & Cov. & Bias & RMSE & Cov. & Rej. & Cov. & Cov. & Rej. \\
\midrule
50 & 1 & -0.006 & 0.915 & 0.012 & 0.264 & 0.915 & 0.870 & 0.915 & 0.955 & 0.095 \\
100 & 1 & -0.007 & 0.935 & -0.002 & 0.176 & 0.925 & 1.000 & 0.950 & 0.950 & 0.135 \\
250 & 1 & -0.013 & 0.930 & 0.017 & 0.110 & 0.935 & 1.000 & 0.950 & 0.935 & 0.240 \\
500 & 1 & 0.001 & 0.960 & -0.005 & 0.070 & 0.940 & 1.000 & 0.960 & 0.950 & 0.395 \\
50 & 2 & -0.005 & 0.925 & -0.007 & 0.304 & 0.965 & 0.665 & 0.965 & 0.945 & 0.070 \\
100 & 2 & 0.006 & 0.905 & -0.006 & 0.226 & 0.940 & 0.975 & 0.975 & 0.945 & 0.075 \\
250 & 2 & 0.003 & 0.910 & 0.002 & 0.127 & 0.965 & 1.000 & 0.930 & 0.965 & 0.275 \\
500 & 2 & -0.009 & 0.915 & -0.005 & 0.109 & 0.920 & 1.000 & 0.905 & 0.945 & 0.395 \\
\bottomrule
\end{tabular}}
\end{table}

\paragraph{Mediation, empirical scenario} Table~\ref{tab:mc_mediation} shows that the plug-in estimators of $\mathrm{NDE}_m$ and $\mathrm{NIE}_m$ are approximately unbiased. Across the eight cells, the absolute bias of $\mathrm{NIE}_m$ is at most 0.017 (true value 0.81), which is within about two MCSE, and the RMSE decreases with $n$ roughly as $n^{-1/2}$, from 0.264 at $n=50$ to 0.070 at $n=500$ with $k=1$. Doubling the residual scale increases the RMSE by a factor between 1.15 and 1.6. Coverage of the percentile intervals ranges from 0.905 to 0.975 over the 32 combinations of effect and cell. Ten values are below 0.93, mostly for $\mathrm{NDE}_m$ with $k=2$ and for $n=50$ with $k=1$; the lowest, 0.905, is about 2.2 MCSE below 0.95. This indicates a mild undercoverage of the percentile interval rather than a systematic failure, and no clear trend with $n$. With $k=1$, power for $\mathrm{NIE}_m$ is 0.87 at $n=50$ and reaches 1 from $n=100$; for $\mathrm{NIE}_l$ (true value 0.09) it increases from 0.375 at $n=50$ to 1 at $n=500$. The indirect effect on the right spread is small (true value 0.03) and is detected in at most 40\% of the replications even at $n=500$.

\paragraph{Crisp estimators of the indirect effect} Table~\ref{tab:mc_centers} answers (Q2). In the empirical scenario, the IM estimator, which uses all components of PK but only the center of CI, has a bias between $-0.130$ and $-0.166$, that is, 16 to 20\% of the true $\mathrm{NIE}_m$, which does not decrease with $n$. In the no-spread-path scenario the same estimator is approximately unbiased (absolute bias at most 0.025, within about two MCSE), whereas the CO estimator has a positive bias between 0.014 and 0.039. The bias of IM in the empirical scenario is therefore due to the omitted paths through the spreads of the mediator, while the omission of the spreads of PK produces a bias of opposite sign. In CO the two omissions partly offset each other, which explains its smaller bias in the empirical scenario, between $-0.103$ and $-0.139$ (13 to 17\%); this is a cancellation of errors, not an advantage. The fuzzy estimator is approximately unbiased in both scenarios, with interval coverage between 0.91 and 0.975 also in the no-spread-path scenario. The crisp estimators have smaller variance: in the no-spread-path scenario IM has a lower RMSE than the fuzzy estimator at every $n$, and in the empirical scenario the fuzzy estimator has the lowest RMSE from $n=250$ (0.070 against 0.124 for CO and 0.151 for IM at $n=500$ with $k=1$), while at $n=100$ its RMSE is similar to that of CO. The fuzzy decomposition is therefore needed when the mediator spreads carry part of the indirect effect; when they do not, a crisp estimator that keeps all components of the predictor is unbiased and more efficient.

\begin{table}[htbp]
\centering
\caption{Monte Carlo study: bias and RMSE of three estimators of the indirect effect on the center of AP ($R=200$). Fuzzy: plug-in estimator of Definition 1. IM: product of coefficients with all components of PK and the center of CI. CO: product of coefficients with the centers of PK and CI only.}
\label{tab:mc_centers}
\small
\resizebox{\textwidth}{!}{%
\begin{tabular}{rrc|cc|cc|cc}
\toprule
 & & & \multicolumn{2}{c|}{Fuzzy} & \multicolumn{2}{c|}{IM} & \multicolumn{2}{c}{CO} \\
$n$ & $k$ & True NIE$_m$ & Bias & RMSE & Bias & RMSE & Bias & RMSE \\
\midrule
\multicolumn{9}{l}{\emph{Empirical scenario: part of the indirect effect passes through the spreads of CI}} \\
\quad 50 & 1 & 0.809 & 0.012 & 0.264 & -0.137 & 0.236 & -0.109 & 0.216 \\
\quad 100 & 1 & 0.809 & -0.002 & 0.176 & -0.136 & 0.179 & -0.107 & 0.158 \\
\quad 250 & 1 & 0.809 & 0.017 & 0.110 & -0.130 & 0.156 & -0.103 & 0.134 \\
\quad 500 & 1 & 0.809 & -0.005 & 0.070 & -0.141 & 0.151 & -0.111 & 0.124 \\
\quad 50 & 2 & 0.815 & -0.007 & 0.304 & -0.166 & 0.315 & -0.139 & 0.284 \\
\quad 100 & 2 & 0.815 & -0.006 & 0.226 & -0.163 & 0.248 & -0.140 & 0.227 \\
\quad 250 & 2 & 0.815 & 0.002 & 0.127 & -0.146 & 0.184 & -0.122 & 0.162 \\
\quad 500 & 2 & 0.815 & -0.005 & 0.109 & -0.155 & 0.181 & -0.134 & 0.161 \\
\multicolumn{9}{l}{\emph{Scenario without spread paths: the indirect effect passes only through the center of CI}} \\
\quad 50 & 1 & 0.643 & 0.000 & 0.257 & 0.007 & 0.194 & 0.030 & 0.204 \\
\quad 100 & 1 & 0.643 & 0.003 & 0.168 & 0.009 & 0.118 & 0.039 & 0.123 \\
\quad 250 & 1 & 0.643 & 0.000 & 0.099 & 0.000 & 0.073 & 0.028 & 0.078 \\
\quad 500 & 1 & 0.643 & -0.001 & 0.077 & -0.007 & 0.053 & 0.020 & 0.057 \\
\quad 50 & 2 & 0.643 & 0.004 & 0.345 & 0.025 & 0.300 & 0.039 & 0.287 \\
\quad 100 & 2 & 0.643 & -0.024 & 0.210 & -0.005 & 0.177 & 0.020 & 0.170 \\
\quad 250 & 2 & 0.643 & -0.008 & 0.132 & -0.008 & 0.119 & 0.014 & 0.116 \\
\quad 500 & 2 & 0.643 & -0.002 & 0.089 & -0.002 & 0.073 & 0.018 & 0.077 \\
\bottomrule
\end{tabular}}
\end{table}

\begin{table}[htbp]
\centering
\caption{Monte Carlo study. Left: empirical size of the bootstrap interval for the indirect effects under the null scenario (no effect of PK on the mediator; nominal 0.05). Right: mean out-of-sample Yang--Ko distance ($n_{\text{test}}=1000$) of the full fuzzy model and of the centers-only model with spreads predicted by their training mean, empirical scenario.}
\label{tab:mc_size_pred}
\small
\begin{tabular}{rrccc|cc}
\toprule
 & & \multicolumn{3}{c|}{Size, null scenario} & \multicolumn{2}{c}{Out-of-sample distance} \\
$n$ & $k$ & NIE$_m$ & NIE$_l$ & NIE$_r$ & Fuzzy & Centers only \\
\midrule
50 & 1 & 0.055 & 0.020 & 0.010 & 10.057 & 9.339 \\
100 & 1 & 0.065 & 0.030 & 0.005 & 9.033 & 8.849 \\
250 & 1 & 0.080 & 0.070 & 0.005 & 8.600 & 8.677 \\
500 & 1 & 0.085 & 0.085 & 0.045 & 8.454 & 8.609 \\
50 & 2 & 0.055 & 0.035 & 0.005 & 39.985 & 37.221 \\
100 & 2 & 0.075 & 0.070 & 0.035 & 36.229 & 35.437 \\
250 & 2 & 0.050 & 0.045 & 0.010 & 34.429 & 34.609 \\
500 & 2 & 0.070 & 0.055 & 0.040 & 33.725 & 34.219 \\
\bottomrule
\end{tabular}
\end{table}

\paragraph{Null scenario and prediction} Under the null scenario (Table~\ref{tab:mc_size_pred}), the rejection rate of the interval for $\mathrm{NIE}_m$ ranges from 0.050 to 0.085, with MCSE between 0.015 and 0.020; the largest values (0.080 and 0.085) occur at $n=250$ and $n=500$ with $k=1$. The test for $\mathrm{NIE}_r$ is conservative, with rates between 0.005 and 0.045, and that for $\mathrm{NIE}_l$ ranges from 0.020 to 0.085, conservative at small $n$ and slightly liberal at $n=500$ with $k=1$. Interval tests of an indirect effect are conservative when both of its factors are zero or close to zero. In the data-generating process, the coefficients of the CI components in the right-spread equation of AP are small (Table~\ref{tab:multivariate}), whereas they are clearly nonzero in the center equation, for which the size is close to nominal. Coverage under the null remains between 0.90 and 0.995. For prediction, the fuzzy model has a larger mean test distance than the centers-only model at $n=50$ (10.10 against 9.34 with $k=1$) and $n=100$ (9.03 against 8.85), and a smaller one at $n=250$ (8.60 against 8.68) and $n=500$ (8.45 against 8.61), where it is better in 94.5\% of the replications. The pattern is the same with $k=2$. The gains of the fuzzy model are small, about 2\% of the test distance at $n=500$, and require a sample size well above that of the application.

\paragraph{Projection and support violations} The nonnegative projection modifies the conditional mean of a transformed spread whenever the linear predictor plus error is negative. In large samples generated from the data-generating processes, this occurred for none of the transformed spreads in the empirical scenario with $k=1$, for 0.26\% of the spreads of CI and 0.08\% of those of AP with $k=2$, and for 0.3\% of the spreads of AP$^*$ in the plasmode experiment. It was more frequent for the spreads of the mediator in the null scenario, 4.0\% with $k=1$ and 18.5\% with $k=2$, because removing the PK paths lowers the mean of the transformed spreads; even there, the size of the indirect-effect tests remained between 0.050 and 0.085. In the empirical scenario, 11.6 to 12.0\% of the generated units with $k=1$ and 23.5 to 24.3\% with $k=2$ had a center or a support outside the admissible range of CI or AP; under the null scenario the proportion was at most 1.2\%. The Gaussian data-generating process therefore produces a non-negligible share of values that cannot occur on bounded rating scales, and the results above refer to the unbounded model on which the estimators are based rather than to the bounded scales of the application.

\begin{table}[htbp]
\centering
\caption{Plasmode Monte Carlo study for moderation. Bias (RMSE) of the interaction coefficients and rejection rate of the Freedman--Lane test ($P=199$ permutations, level 0.05) for the SX and MOT interaction blocks. Planted values: $\theta_{m,\mathrm{SX}}=0.1$, $\theta_{r,\mathrm{SX}}=-0.01$, $\theta_{m,\mathrm{MOT}}=0.03$. $R$: valid replications.}
\label{tab:mc_moderation}
\small
\resizebox{\textwidth}{!}{%
\begin{tabular}{rcccccc}
\toprule
$n$ & SX, center & SX, $\log$ right & MOT, center & Rej. SX & Rej. MOT & $R$ \\
\midrule
\multicolumn{7}{l}{\emph{Planted interaction effects}} \\
\quad 50 & 0.0013 (0.0410) & -0.0002 (0.0036) & 0.0002 (0.0130) & 0.645 & 0.600 & 200 \\
\quad 100 & -0.0007 (0.0263) & -0.0001 (0.0024) & -0.0005 (0.0078) & 0.930 & 0.960 & 200 \\
\quad 250 & -0.0008 (0.0138) & 0.0001 (0.0015) & 0.0008 (0.0053) & 1.000 & 1.000 & 200 \\
\quad 500 & 0.0018 (0.0112) & 0.0001 (0.0010) & 0.0001 (0.0036) & 1.000 & 1.000 & 200 \\
\multicolumn{7}{l}{\emph{No interaction (null)}} \\
\quad 50 & -0.0013 (0.0372) & 0.0004 (0.0044) & -0.0010 (0.0131) & 0.055 & 0.045 & 200 \\
\quad 100 & 0.0010 (0.0274) & -0.0001 (0.0027) & -0.0009 (0.0069) & 0.045 & 0.050 & 200 \\
\quad 250 & -0.0002 (0.0143) & 0.0001 (0.0016) & -0.0003 (0.0044) & 0.030 & 0.045 & 200 \\
\quad 500 & -0.0004 (0.0099) & 0.0001 (0.0011) & 0.0003 (0.0031) & 0.025 & 0.050 & 200 \\
\bottomrule
\end{tabular}}
\end{table}

\paragraph{Moderation} In the plasmode experiment (Table~\ref{tab:mc_moderation}), the estimators of the planted interaction coefficients are approximately unbiased: the largest absolute bias is 0.0018 for $(\mathrm{PK}\circ\mathrm{SX})_m$ (1.8\% of the planted 0.10, about 2.3 MCSE) and 0.0008 for $(\mathrm{PK}\circ\mathrm{MOT})_m$ (2.5\% of 0.03), and the RMSE decreases from 0.041 to 0.011 and from 0.013 to 0.004 between $n=50$ and $n=500$. Under the null scenario, the rejection rate of the Freedman--Lane test is between 0.045 and 0.050 for the MOT block and between 0.025 and 0.055 for the SX block, which is slightly conservative at $n=250$ and $n=500$ (0.030 and 0.025, MCSE about 0.011). Under the planted scenario, power is 0.645 (SX) and 0.600 (MOT) at $n=50$, 0.930 and 0.960 at $n=100$, and 1 from $n=250$. Since the planted effects are of the same order as the empirical estimates, these figures indicate that interactions of that size are detectable with high probability at the sample size of the application, but they do not validate the empirical findings, which depend on the observed data and on the selection of two blocks among four.

\section{Real data application}\label{sec:application}

\subsection{Data}\label{subsec:data}
The application uses data on $n=91$ students of a university Computer Science course, collected in 2014, with no missing values. Six variables are analyzed; three of them were obtained from a questionnaire completed by the students, which is reproduced in \ref{app:questionnaire}, and three were recorded by the instructor.

The questionnaire introduced the notion of a triangular fuzzy number and asked the students to answer each subjective item by drawing a triangle on a graduated scale and writing its lower endpoint, its core and its upper endpoint; asymmetric triangles, decimal values and exact answers (a degenerate triangle) were allowed. The three questionnaire variables are: prior knowledge (PK), the self-assessed knowledge of computer science before the course, on $[0,100]$ (item B1); motivation (MOT), the self-assessed motivation toward university studies, on $[0,10]$ (item C1); and the course indicator (CI), the overall satisfaction with the usefulness and the programme of the course, on $[0,10]$ (item E1). All students chose integer values, so all centers and spreads of PK, MOT and CI are integers. The questionnaire also recorded age, the final high-school grade and the self-assessed understanding of the course content (item D1), which are not used in this paper.

The three variables recorded by the instructor are the following. Academic performance (AP) is the course grade on the Italian scale $[0,30]$, expressed by the instructor as a TFN that summarizes the judgments obtained by the student in two midterm tests and in the final examination: the core is the overall grade, and the spreads express the instructor's imprecision about it in each direction. Course attendance (CR) is a crisp variable computed by the instructor from the attendance records at the end of the course, on a scale from 1 to 8. Sex (SX) is a crisp binary variable, with SX $=1$ for male students (38) and SX $=0$ for female students (53).

All observed supports lie within the corresponding bounds. The spreads have little variation: the left spread of AP equals 2 for 73 students and the right spread for 76, and one right spread is zero for a student with center 30, at the upper bound of the scale. Table~\ref{tab:variables} summarizes the variables and Figure~\ref{fig:tfn} displays the observed triangles.

\begin{table}[htbp]
\centering
\caption{Observed variables ($n=91$). For triangular fuzzy numbers, the table reports the range, mean and standard deviation of the centers and the mean [min, max] of the left and right spreads.}
\label{tab:variables}
\small
\resizebox{\textwidth}{!}{%
\begin{tabular}{lllcccc}
\toprule
Var. & Type & Pot. range & Center range & Mean (SD) center & Left spread & Right spread \\
\midrule
AP & TFN & $[0,30]$ & [10, 30] & 19.40 (3.46) & 1.89 [1, 3] & 1.89 [0, 3] \\
PK & TFN & $[0,100]$ & [30, 100] & 66.16 (16.63) & 3.09 [1, 5] & 2.73 [0, 5] \\
CI & TFN & $[0,10]$ & [4, 9] & 6.63 (1.02) & 1.84 [1, 3] & 1.71 [1, 3] \\
MOT & TFN & $[0,10]$ & [1, 10] & 5.65 (1.96) & 0.85 [0, 1] & 0.80 [0, 1] \\
SX & Crisp binary & $\{0,1\}$ & 0 / 1 & 0.42 (0.50) & -- & -- \\
CR & Crisp & $[1,8]$ & [1, 8] & 4.58 (1.76) & -- & -- \\
\bottomrule
\end{tabular}}
\end{table}

\begin{figure}[htbp]
\centering
\includegraphics[width=0.72\textwidth]{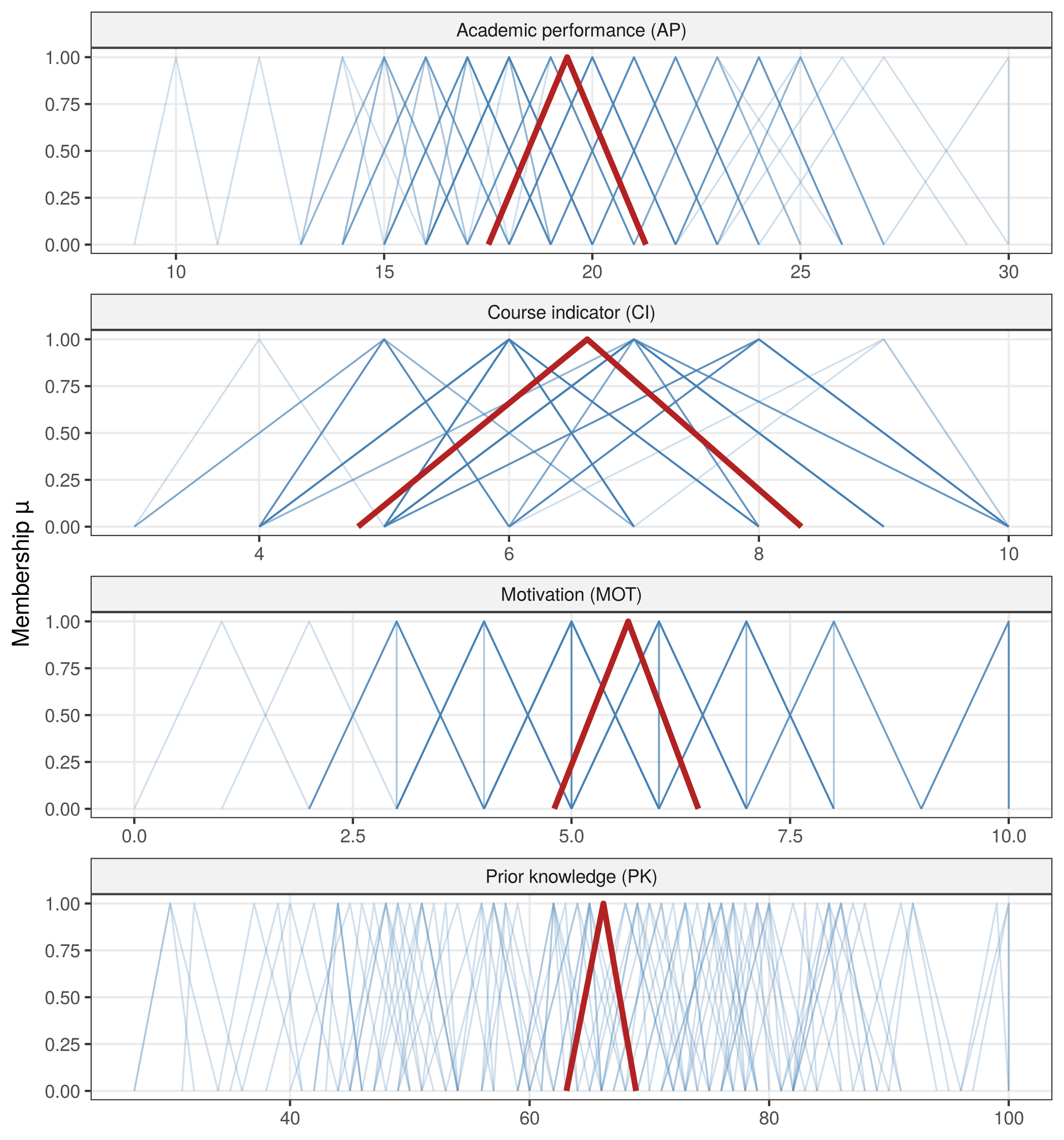}
\caption{Observed triangular fuzzy numbers for academic performance (AP), the course indicator (CI), motivation (MOT) and prior knowledge (PK) (light lines), and componentwise mean triangle (dark line).}
\label{fig:tfn}
\end{figure}

\subsection{Statistical analysis}\label{subsec:analysis}
We estimated the bivariate model of AP on PK and the multivariate model with PK, CI and CR. Coefficient intervals are 95\% percentile case-bootstrap intervals with $B=1000$. Models were compared by $R^2_{\Tsc}$ and by leave-one-out distances on both scales against an intercept-only model, a centers-only model with spreads predicted by the sample mean, and a model with fuzzy response and covariate centers only. The mediation analysis decomposes the association between PK and AP through CI, adjusting for CR, by Definition~\ref{def:effects} with $\delta=10$ points in the center of PK. A sensitivity analysis truncates the counterfactual mediator to the scale of CI and, in a second version, also excludes the students whose shifted PK support leaves $[0,100]$. Four interaction blocks were then added, one at a time, to the multivariate model: PK$\circ$CR (crisp numerical moderator), PK$\circ$CI (fuzzy), PK$\circ$SX (crisp binary, with SX as main effect) and PK$\circ$MOT (fuzzy, with MOT as main effect). Each block was tested with the Freedman--Lane test with $P=999$ permutations, so that the smallest attainable $p$-value is 0.001, and the four $p$-values were adjusted with the Holm procedure. These hypotheses were not specified when the data were collected, and the moderation analyses are exploratory. Interaction regressors were not centered; this does not change fits or tests, but inflates the condition number of the design.

\subsection{Results}\label{subsec:multivariate}

\begin{table}[htbp]
\centering
\caption{Multivariate model of AP on PK, CI and CR: componentwise estimates with 95\% percentile bootstrap intervals ($B=1000$). $R^2_{\mathcal{T}}=0.765$, adjusted $R^2_{\mathcal{T}}=0.746$. Bold: interval excluding zero.}
\label{tab:multivariate}
\small
\begin{tabular}{lccc}
\toprule
Term & Center & $\log$ left spread & $\log$ right spread \\
\midrule
Intercept & 1.852 {\scriptsize [-0.748, 4.316]} & 0.404 {\scriptsize [0.154, 0.650]} & 0.629 {\scriptsize [0.353, 0.871]} \\
$\mathrm{PK}_{m}$ & \textbf{0.074} {\scriptsize [0.035, 0.109]} & 0.001 {\scriptsize [-0.002, 0.004]} & 0.002 {\scriptsize [-0.002, 0.006]} \\
$\mathrm{PK}_{l}$ & 0.126 {\scriptsize [-0.180, 0.409]} & \textbf{0.031} {\scriptsize [0.008, 0.056]} & 0.000 {\scriptsize [-0.030, 0.028]} \\
$\mathrm{PK}_{r}$ & -0.160 {\scriptsize [-0.536, 0.200]} & -0.007 {\scriptsize [-0.036, 0.021]} & 0.019 {\scriptsize [-0.015, 0.062]} \\
$\mathrm{CI}_{m}$ & \textbf{1.652} {\scriptsize [1.189, 2.049]} & \textbf{0.073} {\scriptsize [0.033, 0.115]} & 0.030 {\scriptsize [-0.014, 0.077]} \\
$\mathrm{CI}_{l}$ & 0.345 {\scriptsize [-0.490, 1.204]} & 0.009 {\scriptsize [-0.054, 0.078]} & -0.055 {\scriptsize [-0.154, 0.029]} \\
$\mathrm{CI}_{r}$ & 0.344 {\scriptsize [-0.603, 1.146]} & -0.001 {\scriptsize [-0.065, 0.067]} & 0.055 {\scriptsize [-0.010, 0.120]} \\
$\mathrm{CR}$ & 0.120 {\scriptsize [-0.063, 0.348]} & 0.002 {\scriptsize [-0.010, 0.015]} & 0.009 {\scriptsize [-0.007, 0.028]} \\
\bottomrule
\end{tabular}
\end{table}

\begin{table}[htbp]
\centering
\caption{Model comparison by leave-one-out cross-validated Yang--Ko distance on the original scale (CV$_{\mathcal{O}}$) and on the transformed scale (CV$_{\mathcal{T}}$). Lower is better. $p$ is the number of parameters per component equation.}
\label{tab:model_comparison}
\small
\begin{tabular}{lcccc}
\toprule
Model & $p$ & $R^2_{\mathcal{T}}$ & CV$_{\mathcal{O}}$ & CV$_{\mathcal{T}}$ \\
\midrule
Intercept only (mean TFN) & 1 & 0.000 & 35.970 & 36.180 \\
Centers only, spreads at sample mean & 4 & -- & 9.272 & 9.316 \\
Fuzzy response, covariate centers only & 4 & 0.760 & 9.240 & 9.309 \\
Bivariate fuzzy model (PK) & 4 & 0.610 & 14.954 & 15.097 \\
Full fuzzy model (PK, CI, CR) & 8 & 0.765 & 9.825 & 9.894 \\
\bottomrule
\end{tabular}
\end{table}

\begin{figure}[htbp]
\centering
\includegraphics[width=0.72\textwidth]{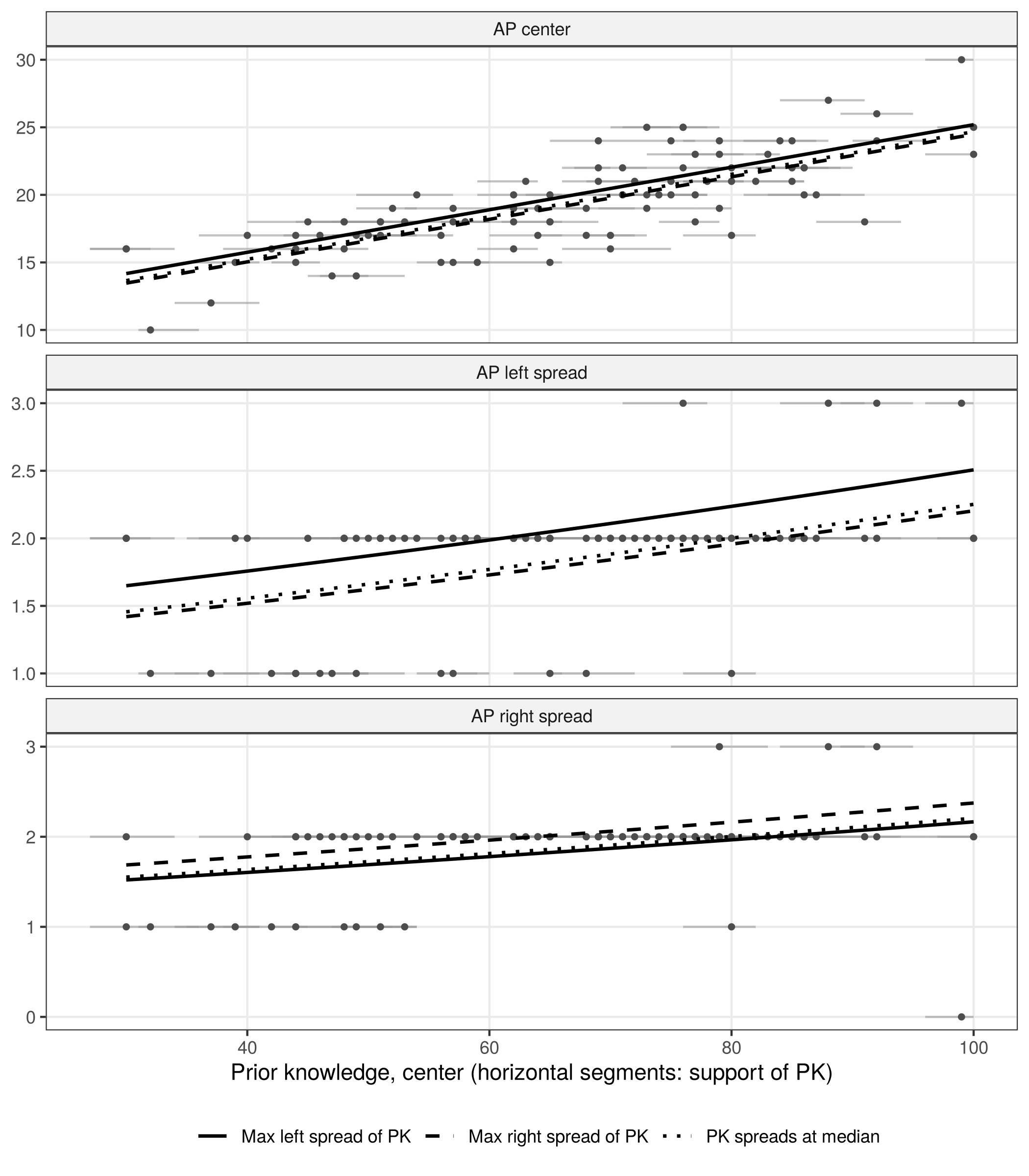}
\caption{Bivariate model of AP on PK. Fitted center, left spread and right spread of AP as functions of the center of PK, with both spreads of PK at their medians (dotted), with the left spread of PK at its maximum (solid) or with the right spread at its maximum (dashed). Points are observed values; horizontal segments are the supports of PK.}
\label{fig:bivariate}
\end{figure}

\paragraph{Regression} In the bivariate model (Figure~\ref{fig:bivariate}; coefficients in Table~\ref{tab:bivariate}, \ref{app:tables}), the center of PK is associated with the center of AP (0.157 per point, 95\% interval $[0.130, 0.187]$) and with the transformed left spread of AP ($0.0040$, $[0.0022, 0.0059]$), and the left spread of PK is associated with the transformed left spread of AP ($0.038$, $[0.010, 0.065]$); $R^2_{\Tsc}=0.610$. Students with higher prior knowledge, and students whose prior knowledge was assessed with a wider left spread, tend to have a more imprecise lower side of academic performance. In the multivariate model (Table~\ref{tab:multivariate}), the center of AP is associated with the centers of PK (0.074, $[0.035, 0.109]$) and CI (1.65, $[1.19, 2.05]$); the transformed left spread of AP is associated with the left spread of PK (0.031, $[0.008, 0.056]$) and with the center of CI (0.073, $[0.033, 0.115]$); no covariate has an interval excluding zero in the right-spread equation, and CR has none in any equation. $R^2_{\Tsc}=0.765$ and the adjusted value is 0.746. Variance inflation factors do not exceed 2.95, two units have leverage above $2p/n=0.176$ (maximum 0.239), and the residual correlation between the center and the transformed left spread is 0.28.

The leave-one-out comparison (Table~\ref{tab:model_comparison}) does not favor the covariate spreads. The full fuzzy model has $\mathrm{CV}_{\Osc}=9.83$, against 35.97 for the intercept-only model, but the model with fuzzy response and covariate centers only has $\mathrm{CV}_{\Osc}=9.24$ and the centers-only model 9.27. The six spread regressors of PK and CI increase $R^2_{\Tsc}$ from 0.760 to 0.765 and worsen the leave-one-out distance by about 6\%. In this sample, therefore, the spreads of the covariates describe associations with the imprecision of AP but do not improve prediction. This agrees with the simulation, in which the fuzzy model was outperformed by the centers-only model at $n=50$ and $n=100$.

\begin{table}[htbp]
\centering
\caption{Mediator equations: CI on PK and CR, componentwise estimates with 95\% percentile bootstrap intervals. $R^2_{\mathcal{T}}=0.480$. Bold: interval excluding zero.}
\label{tab:mediator_eq}
\small
\begin{tabular}{lccc}
\toprule
Term & Center & $\log$ left spread & $\log$ right spread \\
\midrule
Intercept & 3.435 {\scriptsize [2.700, 4.239]} & 0.187 {\scriptsize [0.008, 0.369]} & 0.230 {\scriptsize [0.018, 0.445]} \\
$\mathrm{PK}_{m}$ & \textbf{0.039} {\scriptsize [0.030, 0.049]} & \textbf{0.008} {\scriptsize [0.006, 0.011]} & \textbf{0.008} {\scriptsize [0.006, 0.010]} \\
$\mathrm{PK}_{l}$ & 0.084 {\scriptsize [-0.073, 0.237]} & \textbf{0.038} {\scriptsize [0.006, 0.073]} & -0.026 {\scriptsize [-0.054, 0.005]} \\
$\mathrm{PK}_{r}$ & -0.009 {\scriptsize [-0.167, 0.141]} & 0.025 {\scriptsize [-0.007, 0.054]} & \textbf{0.078} {\scriptsize [0.045, 0.114]} \\
$\mathrm{CR}$ & 0.083 {\scriptsize [-0.001, 0.171]} & \textbf{0.019} {\scriptsize [0.001, 0.039]} & 0.016 {\scriptsize [-0.006, 0.037]} \\
\bottomrule
\end{tabular}
\end{table}

\begin{table}[htbp]
\centering
\caption{Componentwise mediation decomposition for an increase of 10 points in the center of PK, with CI as fuzzy mediator and CR as covariate. Effects are on the original scale of each component of AP. 95\% percentile bootstrap intervals, $B=1000$.}
\label{tab:mediation}
\small
\begin{tabular}{lcc}
\toprule
Effect & Estimate & 95\% interval \\
\midrule
\multicolumn{3}{l}{\emph{Center of AP}} \\
\quad Total (TE) & 1.5444 & [1.2433, 1.8634] \\
\quad Direct (NDE) & 0.7359 & [0.3797, 1.0700] \\
\quad Indirect (NIE) & 0.8085 & [0.4639, 1.1900] \\
\qquad via CI$_m$ & 0.6435 & [0.4167, 0.8858] \\
\qquad via CI$_l$ & 0.0862 & [-0.1264, 0.3237] \\
\qquad via CI$_r$ & 0.0788 & [-0.1104, 0.2888] \\
\qquad via CI$_l$ and CI$_r$ (sum) & 0.1650 & [-0.0897, 0.4904] \\
\multicolumn{3}{l}{\emph{Left spread of AP}} \\
\quad Total (TE) & 0.1140 & [0.0556, 0.1752] \\
\quad Direct (NDE) & 0.0241 & [-0.0643, 0.1157] \\
\quad Indirect (NIE) & 0.0899 & [0.0202, 0.1551] \\
\multicolumn{3}{l}{\emph{Right spread of AP}} \\
\quad Total (TE) & 0.0900 & [-0.0211, 0.1894] \\
\quad Direct (NDE) & 0.0590 & [-0.0487, 0.1675] \\
\quad Indirect (NIE) & 0.0310 & [-0.0587, 0.1142] \\
\bottomrule
\end{tabular}
\end{table}

\begin{figure}[htbp]
\centering
\includegraphics[width=0.85\textwidth]{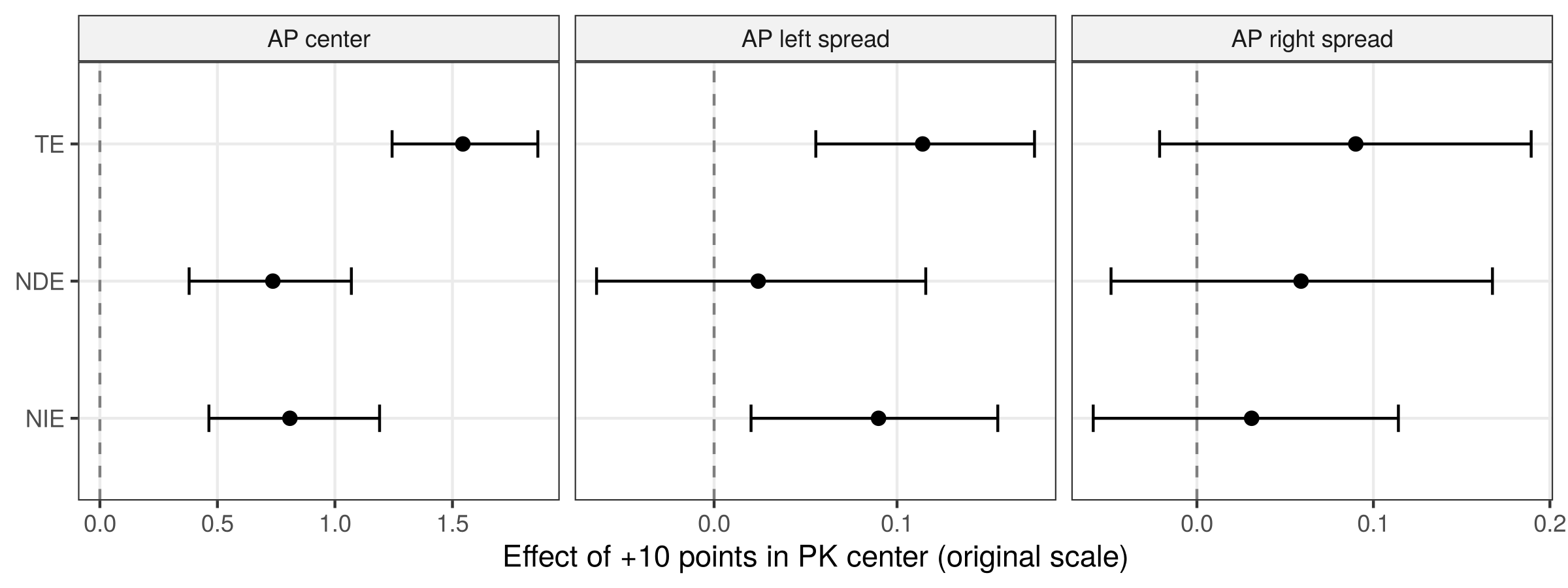}
\caption{Model-based total (TE), direct (NDE) and indirect (NIE) effects of an increase of 10 points in the center of PK on the center and spreads of AP, with CI as fuzzy mediator and CR as covariate. Bars are 95\% percentile case-bootstrap intervals ($B=1000$).}
\label{fig:mediation}
\end{figure}

\paragraph{Mediation}\label{subsec:med_app} In the mediator equations (Table~\ref{tab:mediator_eq}, $R^2_{\Tsc}=0.480$), the center of PK is associated with all three components of CI, and each spread of PK with the corresponding spread of CI (left: 0.038, $[0.006, 0.073]$; right: 0.078, $[0.045, 0.114]$), so that more imprecise assessments of prior knowledge go with more imprecise course indicators. Table~\ref{tab:mediation} and Figure~\ref{fig:mediation} report the decomposition. For the center of AP, a 10-point increase in the center of PK corresponds to $\mathrm{TE}_m=1.54$ $[1.24, 1.86]$, with $\mathrm{NDE}_m=0.74$ $[0.38, 1.07]$ and $\mathrm{NIE}_m=0.81$ $[0.46, 1.19]$, so that the indirect contrast through CI accounts for about 52\% of the model-based total. The path through the center of CI accounts for 0.64 $[0.42, 0.89]$ and equals the product of coefficients, since the center is not transformed; the paths through the left and right spreads of CI (0.09 and 0.08) have intervals that include zero, and so does their sum, 0.165 $[-0.090, 0.490]$. In this sample the indirect contrast through the spreads of CI cannot be distinguished from zero, and the indirect effect on the center is carried mainly by the center of CI. The plug-in total effect differs from the coefficient of the total model, $\delta\widehat{\alpha}_{X_m}=1.53$, by 0.012, as predicted by Proposition~\ref{prop:decomp} for transformed mediator spreads. For the left spread of AP, $\mathrm{TE}_l=0.114$ $[0.056, 0.175]$ and $\mathrm{NIE}_l=0.090$ $[0.020, 0.155]$, while the interval of $\mathrm{NDE}_l$ includes zero; for the right spread no effect has an interval excluding zero. As noted in Section~\ref{subsec:mediation_general}, an interval for a direct effect that includes zero is not evidence of total mediation.

These are decompositions of associations. PK refers to the period before the course but was reported retrospectively in the same questionnaire as CI, AP summarizes tests taken during and at the end of the course, and the conditions for a causal reading are not verifiable. The contrast also requires extrapolation. For 6.6\% of the students the shifted center of PK exceeds 100, and 26.4\% of the students have an observed CI support whose upper end is already 10, so that any positive shift of the mediator moves their counterfactual support beyond the scale (maximum excess 0.73 points). The latter is a consequence of the spread model, which does not impose the upper bound of the scale, and not of the size of $\delta$. Truncating the counterfactual mediator to the scale of CI (Table~\ref{tab:mediation_sensitivity} in \ref{app:tables}) reduces $\mathrm{NIE}_m$ from 0.809 to 0.749 $[0.463, 1.069]$, and further restricting the average to the 83 students whose shifted PK support remains in $[0,100]$ gives 0.759 $[0.466, 1.109]$. $\mathrm{NDE}_m$ is unchanged and $\mathrm{NIE}_l$ remains 0.091 $[0.034, 0.143]$. The truncation mainly removes the small indirect contrast on the right spread (from 0.031 to 0.001) and part of the contribution of the spreads of CI (from 0.165 to 0.105). The conclusions on the center and on the left spread of AP therefore do not depend on the inadmissible counterfactual values.

\begin{table}[htbp]
\centering
\caption{Exploratory interaction blocks for the PK association. $\Delta R^2_{\mathcal{T}}$: in-sample increase over the model without the interaction block; $p$: Freedman--Lane permutation $p$-value ($P=999$); $p_{\mathrm{Holm}}$: Holm adjustment over the four blocks; CV$_{\mathcal{O}}$: leave-one-out distance of the model with and without the interaction block; $\kappa$: condition number of the design matrix.}
\label{tab:interaction_tests}
\small
\resizebox{\textwidth}{!}{%
\begin{tabular}{lcccccc}
\toprule
Interaction block & $\Delta R^2_{\mathcal{T}}$ & $p$ & $p_{\mathrm{Holm}}$ & CV$_{\mathcal{O}}$, with & CV$_{\mathcal{O}}$, without & $\kappa$ \\
\midrule
PK $\circ$ CR (crisp numerical) & 0.003 & 0.800 & 0.800 & 10.343 & 9.825 & 6.1e+03 \\
PK $\circ$ CI (fuzzy) & 0.009 & 0.372 & 0.744 & 10.533 & 9.825 & 1.1e+04 \\
PK $\circ$ SX (crisp binary) & 0.049 & 0.001 & 0.004 & 7.996 & 9.547 & 7.1e+02 \\
PK $\circ$ MOT (fuzzy) & 0.048 & 0.001 & 0.004 & 8.725 & 10.093 & 1.2e+04 \\
\bottomrule
\end{tabular}}
\end{table}

\paragraph{Moderation}\label{subsec:mod_app} Table~\ref{tab:interaction_tests} summarizes the four interaction blocks. The blocks PK$\circ$CR and PK$\circ$CI do not improve the fit ($\Delta R^2_{\Tsc}=0.003$ and 0.009, $p=0.800$ and 0.372) and worsen the leave-one-out distance (from 9.83 to 10.34 and 10.53); the conditional effects of PK at the quartiles of CR and CI have overlapping intervals (Table~\ref{tab:simple_slopes} and Figure~\ref{fig:simple_slopes} in \ref{app:tables}). The blocks PK$\circ$SX and PK$\circ$MOT have $p=0.001$, the smallest attainable value, and Holm-adjusted $p=0.004$; they increase $R^2_{\Tsc}$ by 0.049 and 0.048 and reduce the leave-one-out distance from 9.55 to 8.00 and from 10.09 to 8.73. In the SX model (Table~\ref{tab:moderation_sx} in \ref{app:tables}), the coefficient of $(\mathrm{PK}\circ\mathrm{SX})_m$ in the center equation is 0.090 $[0.041, 0.140]$, so that the slope of the center of AP on the center of PK is 0.042 for female students (SX $=0$) and 0.132 for male students (SX $=1$). In the MOT model (Table~\ref{tab:moderation_mot} in \ref{app:tables}), the coefficients of $(\mathrm{PK}\circ\mathrm{MOT})_m$ and $(\mathrm{PK}\circ\mathrm{MOT})_r$ in the center equation are 0.032 $[0.019, 0.048]$ and 0.072 $[0.013, 0.135]$; by \eqref{eq:cond_fuzzy}, with the spreads of MOT at their median of 1, the point estimate of the slope increases from 0.046 at the first quartile of the center of MOT (4) to 0.141 at the third quartile (7). The association of PK with the center of AP is therefore stronger for male students and for more motivated students. The interaction analyses are exploratory, and the two blocks with small $p$-values are two of the four tested; the condition numbers of the MOT and CI designs exceed $10^4$, and the coefficients of the main effects in these models refer to a moderator equal to zero and are not interpretable on their own.

\section{Discussion and conclusions}\label{sec:discussion}

The paper defines mediation and moderation within the componentwise distance-based regression model for triangular fuzzy data. Three methodological results are relevant for applications. Componentwise least squares minimizes the unconstrained quadratic extension of the Yang--Ko distance on the transformed scale, whatever its shape parameters, but not the distance on the original scale after projection, so goodness of fit and prediction should be assessed on the scale on which they are interpreted, and the closed-form leave-one-out distances make this inexpensive. When the mediator spreads are transformed and carry part of the indirect association, the product-of-coefficients rule no longer holds for the center of the response: the indirect effect depends on the observed spreads, and plug-in contrasts are needed. Effects on the spreads of the response are nonlinear in the parameters and have no product-of-coefficients analogue.

Under the calibrated data-generating process, the simulation supports percentile case-bootstrap intervals for the plug-in effects: bias was negligible, coverage ranged from 0.905 to 0.975, and the size of the indirect-effect tests did not exceed 0.085. Two results limit these conclusions. Indirect effects on the spreads are small relative to their sampling variability, and the effect on the right spread was detected in at most 40\% of the replications at $n=500$. The fuzzy model improved prediction over a centers-only model only from $n=250$, and the improvement was about 2\%; in the application, with $n=91$, the covariate spreads did not improve leave-one-out prediction. Ignoring the paths through the spreads of the mediator, however, biased the indirect effect on the center by 16 to 20\% at every sample size, whereas omitting the spreads of the predictor produced a bias of opposite sign. The fuzzy decomposition is therefore preferable for estimating the indirect effect when the mediator spreads carry part of it, even when it is not preferable for prediction; when they do not, a crisp estimator that keeps all components of the predictor is unbiased and more efficient.

In the application, the indirect contrast through the course indicator accounted for about half of the model-based association between prior knowledge and the center of academic performance, and an indirect contrast was also found for the left spread of academic performance. Prior knowledge was more strongly associated with academic performance among male students and among more motivated students. Truncating the counterfactual mediator to its scale reduced the indirect effect on the center by 7\%, and the contribution of the mediator spreads to that effect was not distinguishable from zero. The mediation analysis has no causal identification, the moderation analyses are exploratory and based on two blocks out of four, and the spread equations are estimated from integer spreads with little variation.

Further limitations concern the model. The $\log(1+\cdot)$ spread model guarantees nonnegative spreads but does not impose the upper bound of bounded rating scales; in the application, counterfactual mediator values fell outside the scale for a quarter of the students (truncating them reduced the indirect effect on the center by 7\%), and in the simulation the Gaussian data-generating process produced inadmissible values for 12 to 24\% of the units. A transformation of the distances of the support endpoints from both bounds would address this. The Dubois--Prade product used for fuzzy interactions is an approximation whose error is small only when the spreads are small relative to the centers; in the application, the error on the support endpoints of the interaction regressors is at most 3.7\% of the center for PK$\otimes$CI and 7.0\% for PK$\otimes$MOT, and only one of the 91 regressors PK$\otimes$MOT has a negative lower endpoint. The Freedman--Lane test is approximate; in the plasmode experiment its size was between 0.025 and 0.055, but it was assessed only for the two blocks of that design, under a correctly specified model with exchangeable residuals, and its level under heteroscedasticity was not studied. The fuzzy mediation and moderation methods of \citet{Yoon2020FuzzyModeration}, \citet{ShinEtAl2026FuzzyMultipleModeration} and \citet{KangEtAl2026BootstrapFuzzyMediation} estimate fuzzy-valued coefficients and indirect effects, whereas the effects of Definition~\ref{def:effects} are real-valued contrasts on each component of the response; since the two approaches estimate different quantities, a numerical comparison would require a common target that neither framework defines. Future work includes bounded spread transformations, sensitivity analysis for unmeasured mediator--outcome confounding, and categorical mediators through the multinomial formulation of Section~\ref{subsec:mediation_fuzzyX_crispZ}.

\section*{Data availability}
The data and the R code that reproduces all tables and figures are available from the corresponding author upon reasonable request.

\section*{Declaration of competing interest}
The authors declare that they have no known competing financial interests or personal relationships that could have appeared to influence the work reported in this paper.

\section*{CRediT authorship contribution statement}
\textbf{F.M.:} Conceptualization, Methodology, Software, Formal analysis, Writing -- original draft. \textbf{A.P.:} Methodology, Data curation, Validation, Writing -- review \& editing.

\section*{Declaration of generative AI and AI-assisted technologies in the manuscript preparation process}
During manuscript preparation, language editing tools were used only to improve readability and English style. The authors reviewed and revised the manuscript and take full responsibility for its scientific content.

\clearpage
\appendix

\section{Diagrams for the remaining configurations}\label{app:figures}
\setcounter{figure}{0}\setcounter{table}{0}
Figure~\ref{fig:app_diagrams} completes the diagrams of Sections~\ref{sec:mediation} and~\ref{sec:moderation} for the configurations not shown in the main text: mediation with a crisp predictor and a fuzzy mediator (panel a), mediation with a fuzzy predictor and a crisp categorical mediator (panel b), and moderation of a fuzzy predictor by a crisp numerical (panel c) or crisp categorical (panel d) moderator. In each panel, the upper diagram corresponds to the absence and the lower diagram to the presence of the effect.

\begin{figure}[H]
\centering
\begin{subfigure}[t]{0.39\textwidth}
\centering
\includegraphics[width=\textwidth]{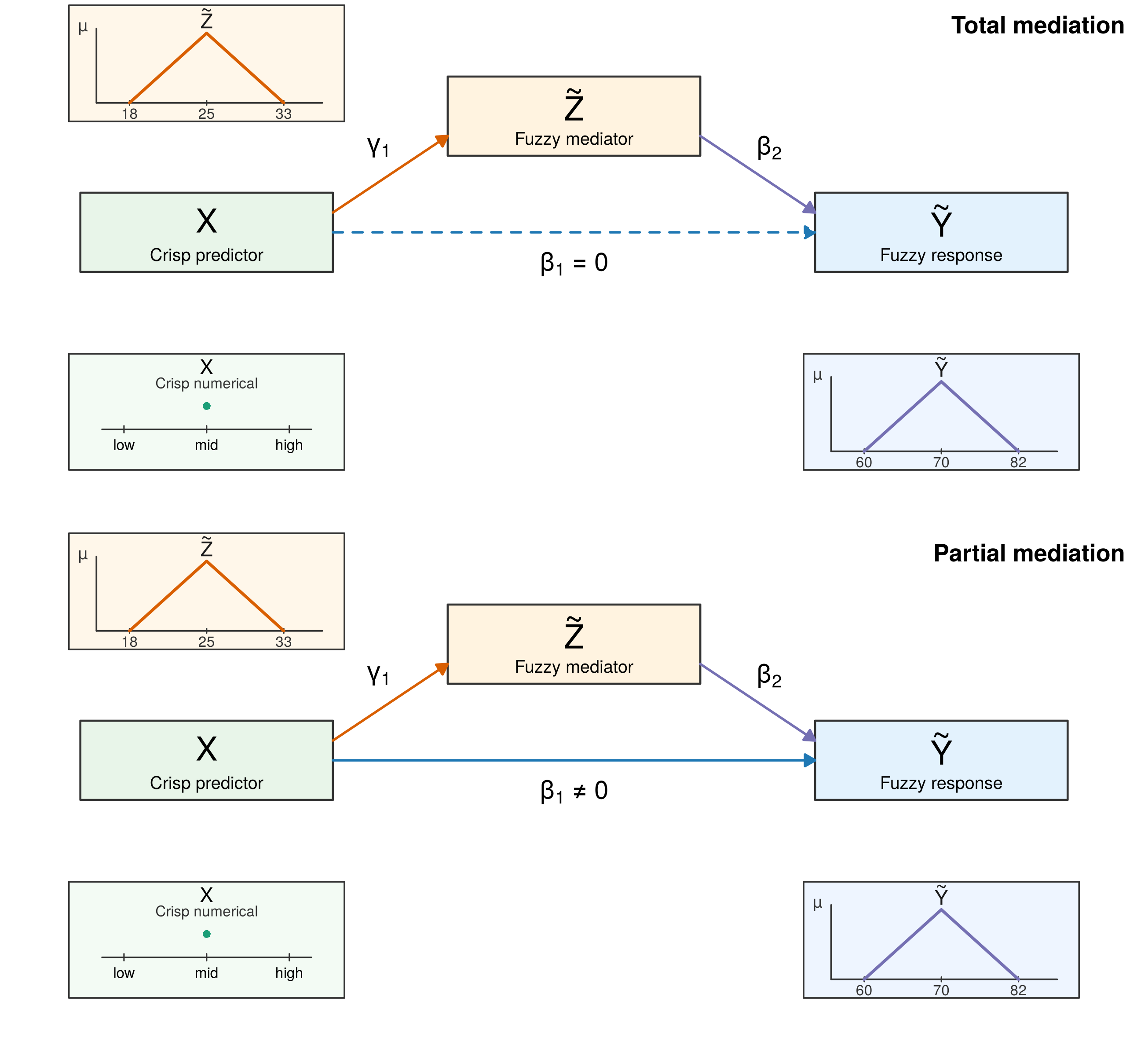}
\caption{Crisp predictor, fuzzy mediator.}
\label{fig:fuzzy_mediation_crispX_total_partial}
\end{subfigure}\hspace{0.04\textwidth}
\begin{subfigure}[t]{0.39\textwidth}
\centering
\includegraphics[width=\textwidth]{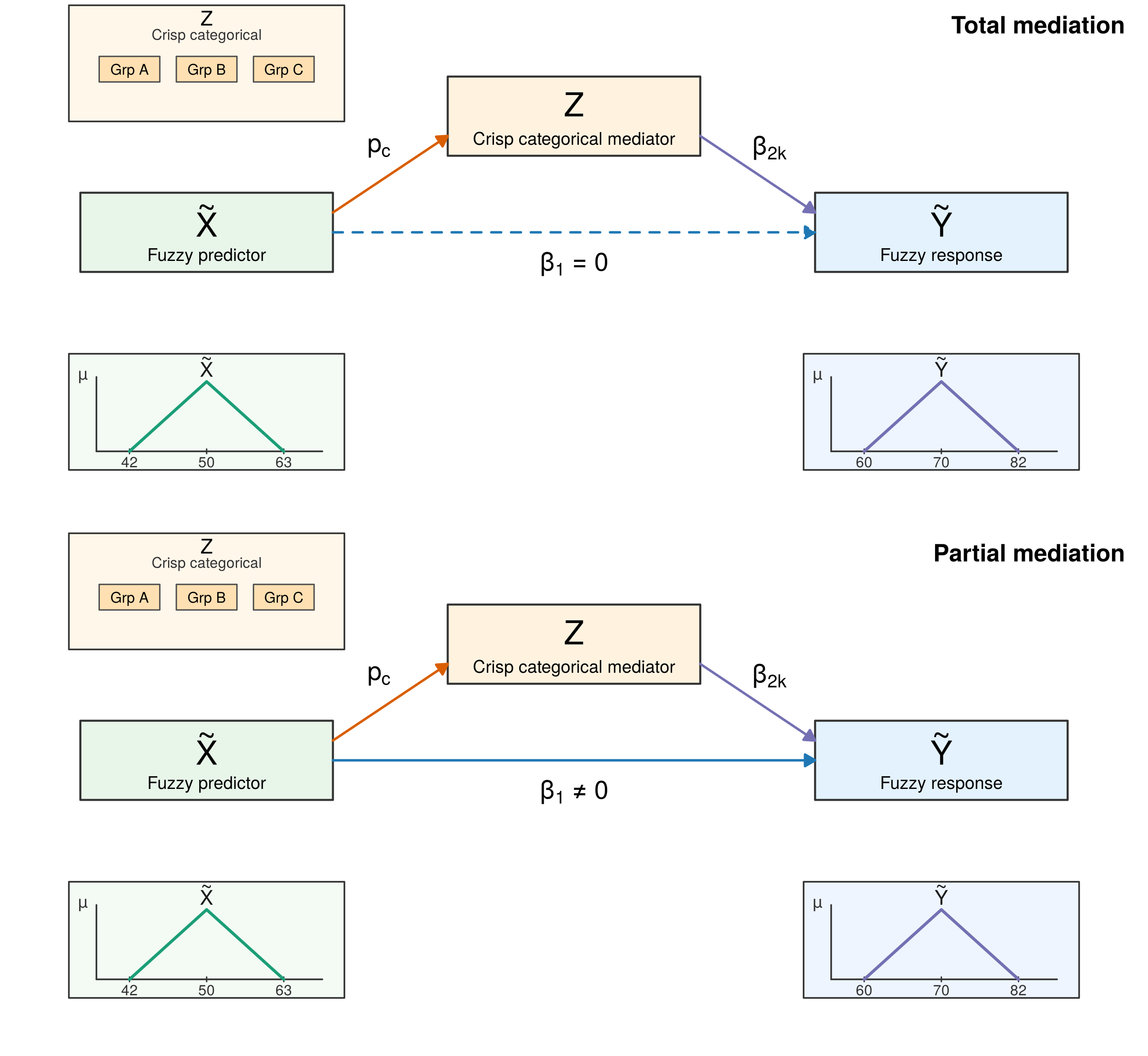}
\caption{Fuzzy predictor, crisp categorical mediator ($p_c$: category probabilities).}
\label{fig:fuzzy_mediation_fuzzyX_qualZ_total_partial}
\end{subfigure}

\vspace{0.6em}
\begin{subfigure}[t]{0.39\textwidth}
\centering
\includegraphics[width=\textwidth]{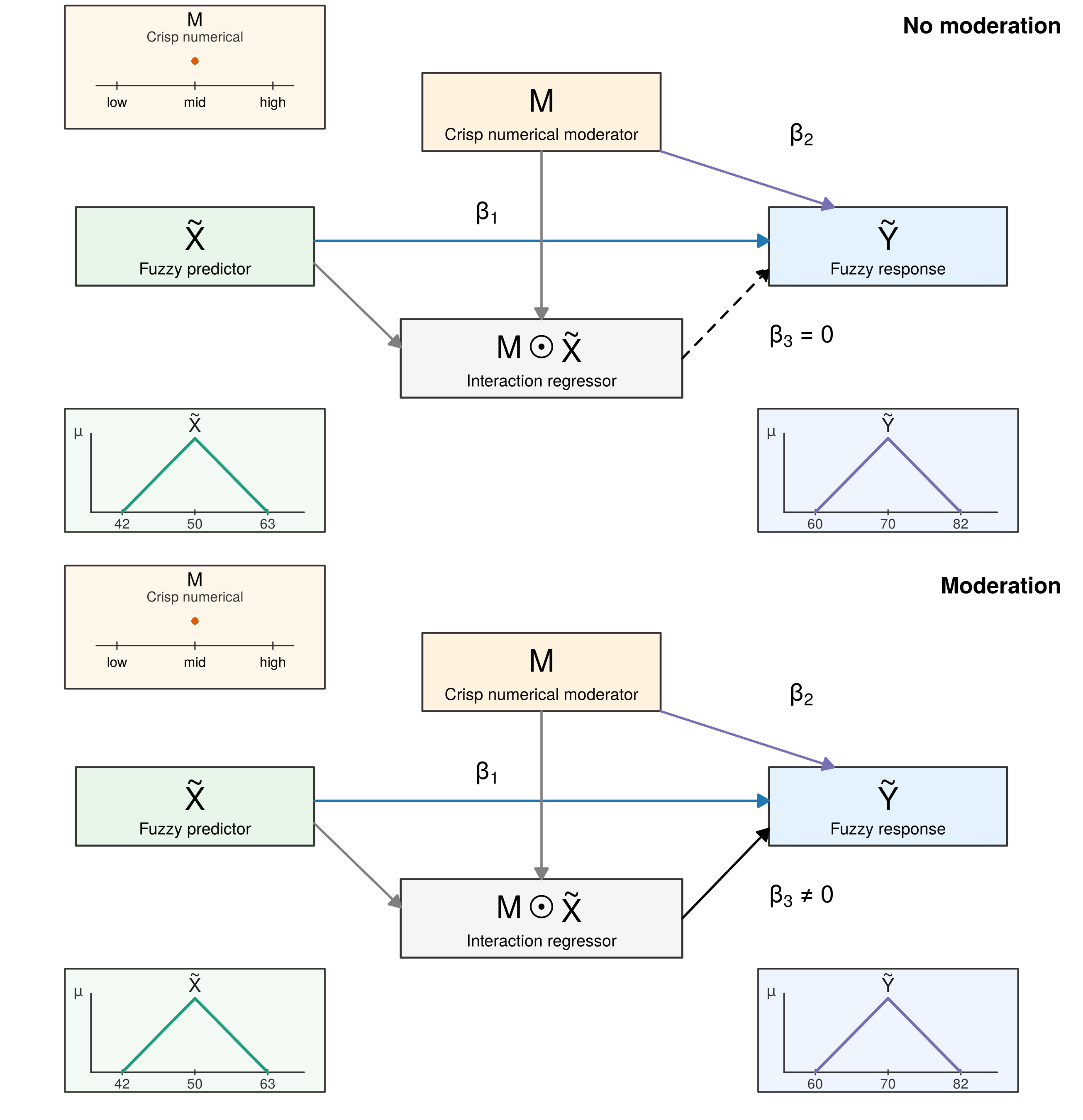}
\caption{Crisp numerical moderator, interaction regressor $M\odot\tilde{X}$: $\beta_3=0$ (upper) and $\beta_3\neq0$ (lower).}
\label{fig:fuzzy_moderation_fuzzy_crisp_num}
\end{subfigure}\hspace{0.04\textwidth}
\begin{subfigure}[t]{0.39\textwidth}
\centering
\includegraphics[width=\textwidth]{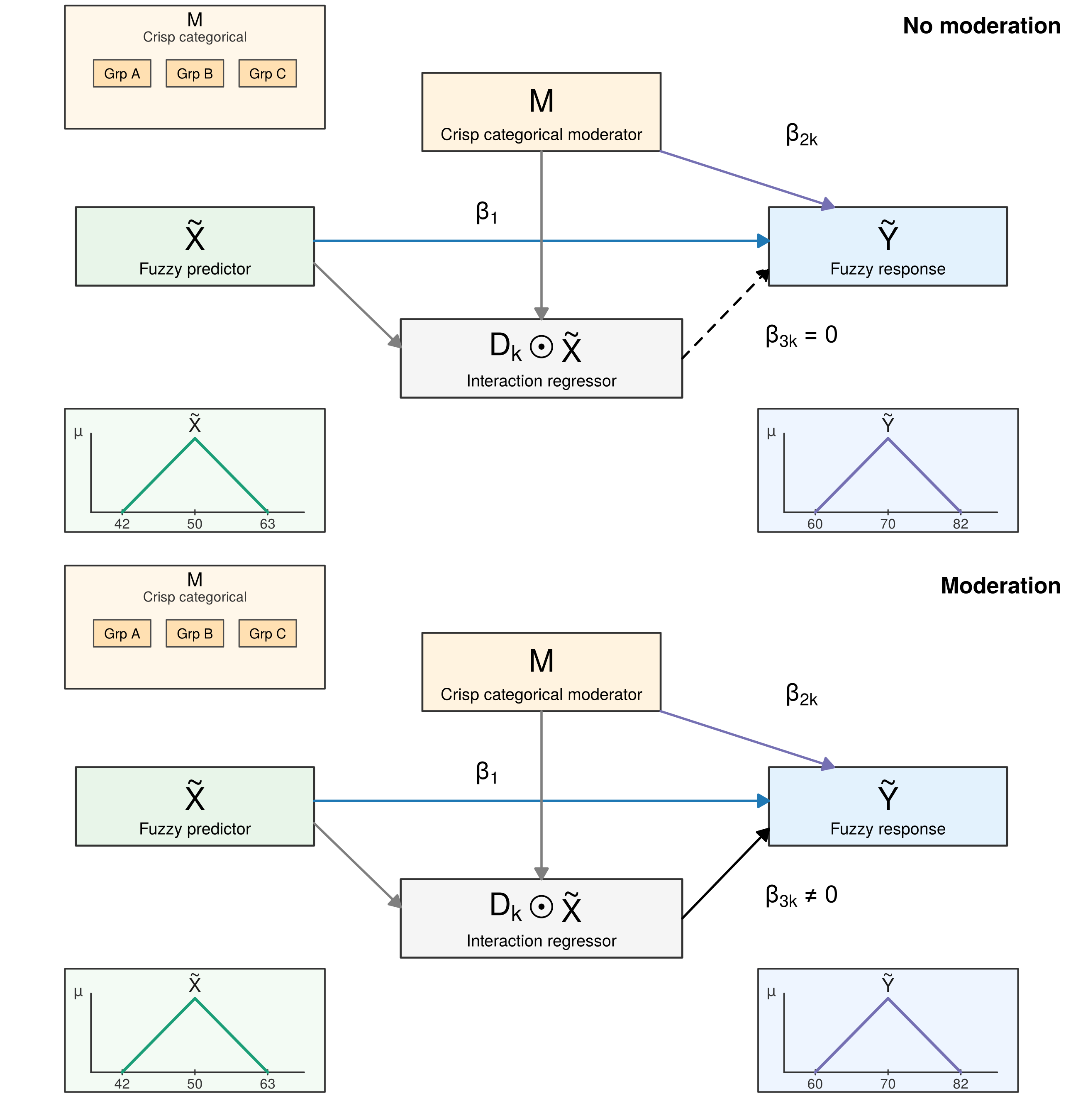}
\caption{Crisp categorical moderator, interaction regressors $D_k\odot\tilde{X}$: all $\beta_{3k}=0$ (upper) and at least one $\beta_{3k}\neq0$ (lower).}
\label{fig:fuzzy_moderation_qualitativeM}
\end{subfigure}
\caption{Mediation and moderation diagrams for configurations with a crisp predictor, mediator or moderator.}
\label{fig:app_diagrams}
\end{figure}

\clearpage
\section{Additional results of the application}\label{app:tables}
\setcounter{figure}{0}\setcounter{table}{0}
This appendix reports the results of Section~\ref{sec:application} that are not shown in the main text: the coefficients of the bivariate model (Table~\ref{tab:bivariate}), the conditional effects of PK for the CR and CI moderation models (Table~\ref{tab:simple_slopes} and Figure~\ref{fig:simple_slopes}), the coefficients of the four interaction models (Tables~\ref{tab:moderation_cr}--\ref{tab:moderation_mot}) and the sensitivity analysis of the mediation decomposition to the scale constraints (Table~\ref{tab:mediation_sensitivity}).

\begin{table}[htbp]
\centering
\caption{Bivariate model of AP on PK: componentwise estimates with 95\% percentile bootstrap intervals ($B=1000$). $R^2_{\mathcal{T}}=0.610$. Bold: interval excluding zero.}
\label{tab:bivariate}
\small
\begin{tabular}{lccc}
\toprule
Term & Center & $\log$ left spread & $\log$ right spread \\
\midrule
Intercept & 8.436 {\scriptsize [6.037, 10.536]} & 0.687 {\scriptsize [0.514, 0.858]} & 0.781 {\scriptsize [0.600, 0.976]} \\
$\mathrm{PK}_{m}$ & \textbf{0.157} {\scriptsize [0.130, 0.187]} & \textbf{0.004} {\scriptsize [0.002, 0.006]} & 0.003 {\scriptsize [-0.001, 0.006]} \\
$\mathrm{PK}_{l}$ & 0.260 {\scriptsize [-0.151, 0.661]} & \textbf{0.038} {\scriptsize [0.010, 0.065]} & -0.006 {\scriptsize [-0.036, 0.022]} \\
$\mathrm{PK}_{r}$ & -0.094 {\scriptsize [-0.501, 0.300]} & -0.007 {\scriptsize [-0.031, 0.017]} & 0.026 {\scriptsize [-0.004, 0.065]} \\
\bottomrule
\end{tabular}
\end{table}

\begin{table}[htbp]
\centering
\caption{Moderation on observed variables. Conditional effect of an increase of 10 points in the center of PK on the center and spreads of AP (original scale), with the moderator at its quartiles and the other covariates at their medians. 95\% percentile bootstrap intervals in brackets. The baseline model without interaction has CV$_{\mathcal{O}}=9.825$.}
\label{tab:simple_slopes}
\small
\resizebox{\textwidth}{!}{%
\begin{tabular}{lccc}
\toprule
Moderator level & AP center & AP left spread & AP right spread \\
\midrule
\multicolumn{4}{l}{\emph{PK $\circ$ CR (fuzzy $\times$ crisp numerical)}} \\
\multicolumn{4}{l}{\quad $\Delta R^2_{\mathcal{T}}=0.003$, permutation $p=0.800$, CV$_{\mathcal{O}}=10.343$} \\
\quad CR at Q1 & 0.701 {\scriptsize [0.251, 1.076]} & 0.015 {\scriptsize [-0.088, 0.125]} & 0.057 {\scriptsize [-0.091, 0.188]} \\
\quad CR at Median & 0.747 {\scriptsize [0.356, 1.108]} & 0.013 {\scriptsize [-0.077, 0.107]} & 0.067 {\scriptsize [-0.054, 0.190]} \\
\quad CR at Q3 & 0.764 {\scriptsize [0.372, 1.176]} & 0.012 {\scriptsize [-0.080, 0.112]} & 0.071 {\scriptsize [-0.044, 0.196]} \\
\multicolumn{4}{l}{\emph{PK $\circ$ CI (fuzzy $\times$ fuzzy)}} \\
\multicolumn{4}{l}{\quad $\Delta R^2_{\mathcal{T}}=0.009$, permutation $p=0.372$, CV$_{\mathcal{O}}=10.533$} \\
\quad CI at Q1 & 0.913 {\scriptsize [0.382, 1.318]} & 0.078 {\scriptsize [-0.022, 0.191]} & 0.025 {\scriptsize [-0.122, 0.175]} \\
\quad CI at Median & 0.790 {\scriptsize [0.379, 1.248]} & 0.042 {\scriptsize [-0.033, 0.153]} & 0.026 {\scriptsize [-0.120, 0.161]} \\
\quad CI at Q3 & 0.790 {\scriptsize [0.359, 1.236]} & 0.042 {\scriptsize [-0.050, 0.142]} & 0.026 {\scriptsize [-0.131, 0.162]} \\
\bottomrule
\end{tabular}}
\end{table}

\begin{figure}[!htbp]
\centering
\includegraphics[width=0.78\textwidth]{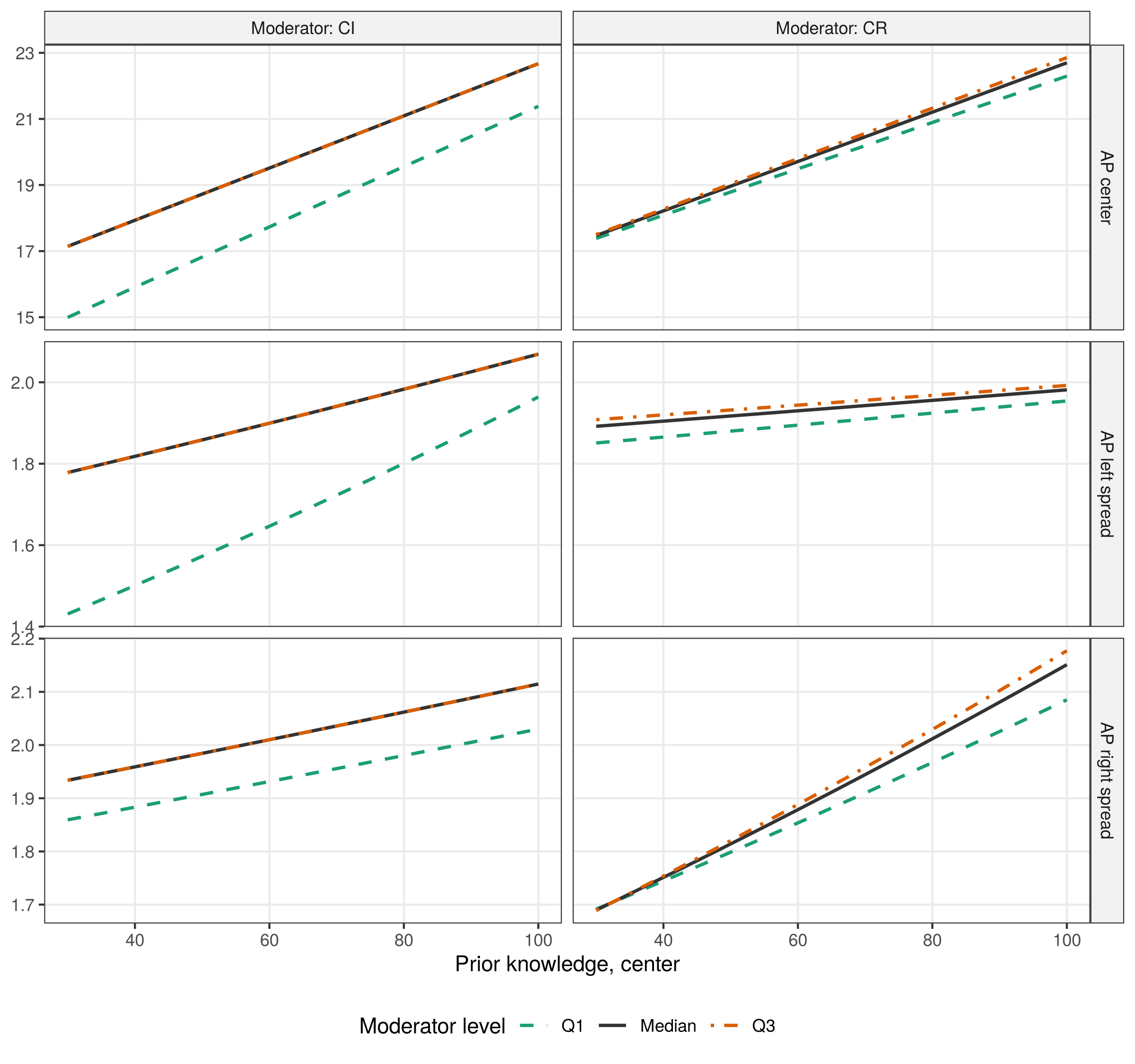}
\caption{Fitted center and spreads of AP as functions of the center of PK, with the moderator (center of CI, left; CR, right) at its first quartile, median and third quartile and the other covariates at their medians. The median and the third quartile of the center of CI are both equal to 7, so the corresponding curves coincide.}
\label{fig:simple_slopes}
\end{figure}

\begin{table}[htbp]
\centering
\caption{Moderation model PK $\circ$ CR (fuzzy $\times$ crisp numerical): componentwise estimates with 95\% percentile bootstrap intervals. $R^2_{\mathcal{T}}=0.768$; Freedman--Lane permutation test for the interaction block, $p=0.800$. Bold: interval excluding zero.}
\label{tab:moderation_cr}
\small
\begin{tabular}{lccc}
\toprule
Term & Center & $\log$ left spread & $\log$ right spread \\
\midrule
Intercept & 1.335 {\scriptsize [-4.091, 7.894]} & 0.578 {\scriptsize [-0.037, 1.122]} & 0.624 {\scriptsize [0.046, 1.287]} \\
$\mathrm{PK}_{m}$ & 0.062 {\scriptsize [-0.026, 0.123]} & 0.001 {\scriptsize [-0.005, 0.008]} & 0.001 {\scriptsize [-0.009, 0.008]} \\
$\mathrm{PK}_{l}$ & 0.595 {\scriptsize [-0.306, 1.486]} & 0.018 {\scriptsize [-0.076, 0.112]} & -0.003 {\scriptsize [-0.098, 0.097]} \\
$\mathrm{PK}_{r}$ & -0.293 {\scriptsize [-1.706, 0.767]} & -0.063 {\scriptsize [-0.149, 0.020]} & 0.046 {\scriptsize [-0.062, 0.226]} \\
$\mathrm{CI}_{m}$ & \textbf{1.689} {\scriptsize [1.225, 2.092]} & \textbf{0.074} {\scriptsize [0.035, 0.119]} & 0.029 {\scriptsize [-0.021, 0.069]} \\
$\mathrm{CI}_{l}$ & 0.277 {\scriptsize [-0.545, 1.252]} & 0.015 {\scriptsize [-0.055, 0.083]} & -0.055 {\scriptsize [-0.174, 0.039]} \\
$\mathrm{CI}_{r}$ & 0.378 {\scriptsize [-0.587, 1.328]} & 0.005 {\scriptsize [-0.075, 0.071]} & 0.052 {\scriptsize [-0.016, 0.126]} \\
$(\mathrm{PK}\circ\mathrm{CR})_{m}$ & 0.003 {\scriptsize [-0.009, 0.019]} & 0.000 {\scriptsize [-0.001, 0.001]} & 0.000 {\scriptsize [-0.001, 0.002]} \\
$(\mathrm{PK}\circ\mathrm{CR})_{l}$ & -0.103 {\scriptsize [-0.299, 0.079]} & 0.003 {\scriptsize [-0.018, 0.022]} & 0.000 {\scriptsize [-0.019, 0.019]} \\
$(\mathrm{PK}\circ\mathrm{CR})_{r}$ & 0.025 {\scriptsize [-0.174, 0.252]} & 0.011 {\scriptsize [-0.003, 0.025]} & -0.005 {\scriptsize [-0.033, 0.011]} \\
$\mathrm{CR}$ & 0.207 {\scriptsize [-1.029, 1.182]} & -0.032 {\scriptsize [-0.131, 0.088]} & 0.009 {\scriptsize [-0.110, 0.120]} \\
\bottomrule
\end{tabular}
\end{table}

\begin{table}[htbp]
\centering
\caption{Moderation model PK $\circ$ CI (fuzzy $\times$ fuzzy): componentwise estimates with 95\% percentile bootstrap intervals. $R^2_{\mathcal{T}}=0.775$; Freedman--Lane permutation test for the interaction block, $p=0.372$. Bold: interval excluding zero.}
\label{tab:moderation_ci}
\small
\begin{tabular}{lccc}
\toprule
Term & Center & $\log$ left spread & $\log$ right spread \\
\midrule
Intercept & 3.249 {\scriptsize [-5.780, 16.071]} & 0.195 {\scriptsize [-0.758, 1.108]} & -0.103 {\scriptsize [-1.458, 1.325]} \\
$\mathrm{PK}_{m}$ & 0.059 {\scriptsize [-0.151, 0.213]} & 0.005 {\scriptsize [-0.010, 0.020]} & 0.013 {\scriptsize [-0.007, 0.033]} \\
$\mathrm{PK}_{l}$ & -0.113 {\scriptsize [-0.530, 0.452]} & 0.019 {\scriptsize [-0.018, 0.058]} & 0.022 {\scriptsize [-0.028, 0.072]} \\
$\mathrm{PK}_{r}$ & -0.282 {\scriptsize [-0.798, 0.122]} & -0.026 {\scriptsize [-0.068, 0.007]} & 0.032 {\scriptsize [-0.014, 0.084]} \\
$\mathrm{CI}_{m}$ & 2.359 {\scriptsize [-0.108, 3.898]} & 0.166 {\scriptsize [-0.009, 0.340]} & 0.043 {\scriptsize [-0.180, 0.289]} \\
$\mathrm{CI}_{l}$ & -1.931 {\scriptsize [-5.141, 2.000]} & 0.012 {\scriptsize [-0.282, 0.318]} & 0.228 {\scriptsize [-0.226, 0.702]} \\
$\mathrm{CI}_{r}$ & -0.855 {\scriptsize [-3.534, 2.463]} & -0.233 {\scriptsize [-0.492, 0.061]} & 0.161 {\scriptsize [-0.213, 0.434]} \\
$(\mathrm{PK}\circ\mathrm{CI})_{m}$ & -0.012 {\scriptsize [-0.036, 0.027]} & -0.001 {\scriptsize [-0.004, 0.001]} & 0.000 {\scriptsize [-0.004, 0.003]} \\
$(\mathrm{PK}\circ\mathrm{CI})_{l}$ & 0.035 {\scriptsize [-0.028, 0.089]} & 0.000 {\scriptsize [-0.004, 0.004]} & -0.004 {\scriptsize [-0.011, 0.003]} \\
$(\mathrm{PK}\circ\mathrm{CI})_{r}$ & 0.018 {\scriptsize [-0.028, 0.059]} & 0.003 {\scriptsize [0.000, 0.007]} & -0.002 {\scriptsize [-0.006, 0.004]} \\
$\mathrm{CR}$ & 0.132 {\scriptsize [-0.059, 0.351]} & 0.003 {\scriptsize [-0.008, 0.016]} & 0.008 {\scriptsize [-0.008, 0.026]} \\
\bottomrule
\end{tabular}
\end{table}

\begin{table}[htbp]
\centering
\caption{Moderation model PK $\circ$ SX (fuzzy $\times$ crisp binary): componentwise estimates with 95\% percentile bootstrap intervals. $R^2_{\mathcal{T}}=0.826$; Freedman--Lane permutation test for the interaction block, $p=0.001$. Bold: interval excluding zero.}
\label{tab:moderation_sx}
\small
\begin{tabular}{lccc}
\toprule
Term & Center & $\log$ left spread & $\log$ right spread \\
\midrule
Intercept & 3.632 {\scriptsize [0.842, 6.777]} & 0.587 {\scriptsize [0.324, 0.852]} & 0.713 {\scriptsize [0.390, 0.989]} \\
$\mathrm{PK}_{m}$ & \textbf{0.042} {\scriptsize [0.002, 0.077]} & -0.001 {\scriptsize [-0.004, 0.002]} & 0.001 {\scriptsize [-0.003, 0.005]} \\
$\mathrm{PK}_{l}$ & 0.066 {\scriptsize [-0.291, 0.394]} & 0.016 {\scriptsize [-0.013, 0.043]} & -0.006 {\scriptsize [-0.034, 0.024]} \\
$\mathrm{PK}_{r}$ & 0.015 {\scriptsize [-0.321, 0.393]} & -0.007 {\scriptsize [-0.032, 0.020]} & 0.005 {\scriptsize [-0.032, 0.039]} \\
$\mathrm{CI}_{m}$ & \textbf{1.591} {\scriptsize [1.122, 2.007]} & \textbf{0.072} {\scriptsize [0.030, 0.114]} & 0.034 {\scriptsize [-0.020, 0.087]} \\
$\mathrm{CI}_{l}$ & 0.246 {\scriptsize [-0.594, 1.046]} & 0.005 {\scriptsize [-0.062, 0.070]} & -0.060 {\scriptsize [-0.154, 0.027]} \\
$\mathrm{CI}_{r}$ & 0.451 {\scriptsize [-0.361, 1.274]} & 0.009 {\scriptsize [-0.057, 0.064]} & 0.062 {\scriptsize [-0.005, 0.141]} \\
$(\mathrm{PK}\circ\mathrm{SX})_{m}$ & \textbf{0.090} {\scriptsize [0.041, 0.140]} & \textbf{0.005} {\scriptsize [0.001, 0.009]} & 0.002 {\scriptsize [-0.007, 0.010]} \\
$(\mathrm{PK}\circ\mathrm{SX})_{l}$ & -0.333 {\scriptsize [-1.008, 0.380]} & 0.035 {\scriptsize [-0.019, 0.096]} & 0.029 {\scriptsize [-0.071, 0.139]} \\
$(\mathrm{PK}\circ\mathrm{SX})_{r}$ & -0.491 {\scriptsize [-1.156, 0.207]} & 0.015 {\scriptsize [-0.040, 0.065]} & 0.046 {\scriptsize [-0.040, 0.156]} \\
$\mathrm{CR}$ & 0.132 {\scriptsize [-0.031, 0.324]} & 0.003 {\scriptsize [-0.010, 0.017]} & 0.008 {\scriptsize [-0.009, 0.027]} \\
$\mathrm{SX}$ & -2.843 {\scriptsize [-6.764, 1.247]} & \textbf{-0.533} {\scriptsize [-0.871, -0.168]} & -0.352 {\scriptsize [-0.807, 0.092]} \\
\bottomrule
\end{tabular}
\end{table}

\begin{table}[htbp]
\centering
\caption{Moderation model PK $\circ$ MOT (fuzzy $\times$ fuzzy): componentwise estimates with 95\% percentile bootstrap intervals. $R^2_{\mathcal{T}}=0.824$; Freedman--Lane permutation test for the interaction block, $p=0.001$. Bold: interval excluding zero.}
\label{tab:moderation_mot}
\small
\begin{tabular}{lccc}
\toprule
Term & Center & $\log$ left spread & $\log$ right spread \\
\midrule
Intercept & 18.092 {\scriptsize [9.331, 28.373]} & 1.044 {\scriptsize [0.191, 1.957]} & 0.749 {\scriptsize [-0.755, 2.021]} \\
$\mathrm{PK}_{m}$ & \textbf{-0.164} {\scriptsize [-0.309, -0.048]} & -0.007 {\scriptsize [-0.018, 0.004]} & 0.000 {\scriptsize [-0.016, 0.020]} \\
$\mathrm{PK}_{l}$ & 0.014 {\scriptsize [-0.443, 0.527]} & 0.042 {\scriptsize [-0.004, 0.093]} & 0.027 {\scriptsize [-0.033, 0.097]} \\
$\mathrm{PK}_{r}$ & \textbf{-0.584} {\scriptsize [-1.085, -0.084]} & \textbf{-0.045} {\scriptsize [-0.086, -0.009]} & 0.007 {\scriptsize [-0.058, 0.085]} \\
$\mathrm{CI}_{m}$ & \textbf{1.702} {\scriptsize [1.243, 2.119]} & \textbf{0.082} {\scriptsize [0.034, 0.125]} & 0.033 {\scriptsize [-0.020, 0.075]} \\
$\mathrm{CI}_{l}$ & 0.444 {\scriptsize [-0.268, 1.294]} & 0.007 {\scriptsize [-0.055, 0.072]} & -0.059 {\scriptsize [-0.166, 0.032]} \\
$\mathrm{CI}_{r}$ & 0.062 {\scriptsize [-0.731, 0.798]} & -0.015 {\scriptsize [-0.087, 0.046]} & 0.055 {\scriptsize [-0.011, 0.136]} \\
$\mathrm{MOT}_{m}$ & \textbf{-2.176} {\scriptsize [-3.427, -1.205]} & -0.078 {\scriptsize [-0.190, 0.019]} & -0.053 {\scriptsize [-0.198, 0.124]} \\
$\mathrm{MOT}_{l}$ & -0.939 {\scriptsize [-4.332, 2.698]} & 0.103 {\scriptsize [-0.299, 0.578]} & 0.329 {\scriptsize [-0.071, 0.712]} \\
$\mathrm{MOT}_{r}$ & \textbf{-4.874} {\scriptsize [-8.707, -1.105]} & \textbf{-0.407} {\scriptsize [-0.762, -0.057]} & -0.153 {\scriptsize [-0.634, 0.403]} \\
$(\mathrm{PK}\circ\mathrm{MOT})_{m}$ & \textbf{0.032} {\scriptsize [0.019, 0.048]} & 0.001 {\scriptsize [0.000, 0.002]} & 0.001 {\scriptsize [-0.001, 0.003]} \\
$(\mathrm{PK}\circ\mathrm{MOT})_{l}$ & 0.012 {\scriptsize [-0.048, 0.063]} & -0.002 {\scriptsize [-0.009, 0.004]} & -0.004 {\scriptsize [-0.010, 0.002]} \\
$(\mathrm{PK}\circ\mathrm{MOT})_{r}$ & \textbf{0.072} {\scriptsize [0.013, 0.135]} & \textbf{0.006} {\scriptsize [0.001, 0.011]} & 0.002 {\scriptsize [-0.006, 0.010]} \\
$\mathrm{CR}$ & 0.179 {\scriptsize [-0.027, 0.355]} & 0.004 {\scriptsize [-0.011, 0.019]} & 0.010 {\scriptsize [-0.007, 0.027]} \\
\bottomrule
\end{tabular}
\end{table}

\begin{table}[htbp]
\centering
\caption{Sensitivity of the mediation decomposition to the scale constraints, increase of 10 points in the center of PK. (a) Main analysis; (b) counterfactual mediator truncated to the scale $[0,10]$ of CI; (c) as (b), averaging only over the 83 of 91 students whose shifted PK support remains in $[0,100]$. 95\% percentile bootstrap intervals, $B=1000$.}
\label{tab:mediation_sensitivity}
\small
\resizebox{\textwidth}{!}{%
\begin{tabular}{lccc}
\toprule
Effect & (a) Main & (b) Truncated mediator & (c) Truncated mediator, admissible PK \\
\midrule
TE, center & 1.544 {\scriptsize [1.243, 1.863]} & 1.485 {\scriptsize [1.169, 1.798]} & 1.495 {\scriptsize [1.172, 1.823]} \\
NDE, center & 0.736 {\scriptsize [0.380, 1.070]} & 0.736 {\scriptsize [0.340, 1.070]} & 0.736 {\scriptsize [0.360, 1.076]} \\
NIE, center & 0.809 {\scriptsize [0.464, 1.190]} & 0.749 {\scriptsize [0.463, 1.069]} & 0.759 {\scriptsize [0.466, 1.109]} \\
\quad via CI$_l$ and CI$_r$ & 0.165 {\scriptsize [-0.090, 0.490]} & 0.105 {\scriptsize [-0.109, 0.350]} & 0.115 {\scriptsize [-0.119, 0.345]} \\
TE, left spread & 0.114 {\scriptsize [0.056, 0.175]} & 0.115 {\scriptsize [0.045, 0.189]} & 0.113 {\scriptsize [0.046, 0.185]} \\
NIE, left spread & 0.090 {\scriptsize [0.020, 0.155]} & 0.091 {\scriptsize [0.034, 0.143]} & 0.089 {\scriptsize [0.038, 0.142]} \\
TE, right spread & 0.090 {\scriptsize [-0.021, 0.189]} & 0.060 {\scriptsize [-0.050, 0.168]} & 0.067 {\scriptsize [-0.046, 0.166]} \\
NIE, right spread & 0.031 {\scriptsize [-0.059, 0.114]} & 0.001 {\scriptsize [-0.077, 0.065]} & 0.008 {\scriptsize [-0.061, 0.082]} \\
\bottomrule
\end{tabular}}
\end{table}

\clearpage
\section{Questionnaire}\label{app:questionnaire}
\setcounter{figure}{0}\setcounter{table}{0}
The questionnaire completed by the students in 2014 is reproduced below; only the layout has been adapted to the page. Items B1, C1 and E1 give the variables PK, MOT and CI of Section~\ref{subsec:data}, item A1 gives SX. Age (A2), final high-school grade (A3) and understanding of the course content (D1) were collected but are not used in this paper. Academic performance (AP) and course attendance (CR) are not part of the questionnaire, because they were recorded by the instructor.

\medskip
\noindent\rule{\textwidth}{0.8pt}
\begingroup\small\setlength{\parskip}{4pt}
\begin{center}

\textbf{Student questionnaire on academic background, motivation and course experience}\\
Computer Science course $\mid$ Prof.\ F. Maturo
\end{center}

\textbf{Purpose of the activity.} This questionnaire was designed to record students' subjective assessments without forcing them to choose a single exact number. The information can be used to introduce fuzzy sets, triangular fuzzy numbers and the analysis of imprecise evaluations during the Computer Science course. It also provides a basis for methodological research on associations between educational characteristics and students' perceptions.

\textbf{What is vagueness?} A judgment can be vague even when the question is understood perfectly. For example, you may consider your motivation to be ``around 7 out of 10'' while also accepting somewhat lower or higher values as reasonable descriptions. This is different from random variation: the degree of membership in a fuzzy set is not a probability.

A triangular fuzzy number describes such a judgment using three points: lower endpoint (L), most representative value or core (C), upper endpoint (U). The triangle rises from membership 0 at L to membership 1 at C, then falls to membership 0 at U. The left spread is C $-$ L; the right spread is U $-$ C. Unequal spreads are allowed.

\textbf{Example: an approximate judgment on a 0--10 scale.}

\begin{center}
\begin{tikzpicture}[x=1.15cm,y=1.3cm]
\foreach \t in {0,...,10} {\draw[gray!40,densely dotted] (\t,0)--(\t,1);
  \draw[thick] (\t,-0.06)--(\t,0.06); \node[below,font=\scriptsize] at (\t,-0.06) {\t};}
\draw[gray!60,dashed] (0,1)--(10,1);
\fill[gray!15] (3.2,0)--(6.4,1)--(8.7,0)--cycle;
\draw[very thick] (0,0)--(10,0);
\draw[very thick] (3.2,0)--(6.4,1)--(8.7,0);
\fill (3.2,0) circle (1.5pt); \fill (6.4,1) circle (1.5pt); \fill (8.7,0) circle (1.5pt);
\node[above,font=\scriptsize] at (3.2,0.12) {L = 3.2};
\node[above,font=\scriptsize\bfseries] at (6.4,1) {C = 6.4};
\node[above,font=\scriptsize] at (8.7,0.12) {U = 8.7};
\node[left,font=\scriptsize] at (0,0) {0}; \node[left,font=\scriptsize] at (0,1) {1};
\node[font=\scriptsize] at (-0.6,1.3) {$\mu$};
\end{tikzpicture}
\end{center}
Here L $=3.2$, C $=6.4$ and U $=8.7$; therefore, the left spread is 3.2 and the right spread is 2.3. These are degrees of compatibility with a judgment, not probabilities.

\textbf{How to respond.} For each subjective question, draw a triangle directly on the numbered scale: put its left foot at L, its peak at C (height 1) and its right foot at U. Mark the three numbers in the spaces provided below each scale. You may choose any point along the continuous axis, including values between the labeled ticks or decimal values. Your triangle may be asymmetric. If your answer is exact, L, C and U may coincide.

\medskip
\begin{center}\textbf{PART 1 $\cdot$ YOUR BACKGROUND AND INITIAL ASSESSMENT}\end{center}

\textbf{A. Background information.} These entries describe factual characteristics. No triangle is needed for this section.

\noindent\textbf{A1} Sex: \quad $\square$ Female \quad $\square$ Male \hfill \textbf{A2} Age (years): \underline{\hspace{2.5cm}}\\[6pt]
\textbf{A3} Final high-school grade: \underline{\hspace{2.5cm}} \hfill Survey year: 2014

\medskip
\textbf{B. Prior knowledge of computer science}

\textbf{B1} Before attending this course, how would you describe your knowledge of computer science?
0 $=$ no prior knowledge; 100 $=$ very extensive prior knowledge. Indicate the most representative level and the acceptable range around it.

\begin{center}\qscale{100}{10}\end{center}

\textbf{C. Motivation toward university studies}

\textbf{C1} How would you assess your overall motivation to pursue your university studies?
0 $=$ no motivation; 10 $=$ extremely high motivation. Draw the triangle that best represents your current judgment.

\begin{center}\qscale{10}{1}\end{center}

\medskip
\begin{center}\textbf{PART 2 $\cdot$ YOUR ASSESSMENT OF THE COURSE}\end{center}

\textbf{D. Understanding of course content}

\textbf{D1} How well do you feel you understand the main topics addressed in this Computer Science course?
0 $=$ no understanding; 10 $=$ excellent understanding. Include the range of values that you consider reasonably compatible with your assessment.

\begin{center}\qscale{10}{1}\end{center}

\textbf{E. Evaluation of the course}

\textbf{E1} Considering the usefulness of the course and the appropriateness of its programme, how satisfied are you with the course overall?
0 $=$ not at all satisfied; 10 $=$ completely satisfied. Give one overall assessment of these aspects, indicating both its most representative value and its acceptable range.

\begin{center}\qscale{10}{1}\end{center}

\textbf{Before handing in the form.} Please check that every subjective response includes a marked lower endpoint, a core and an upper endpoint, with L $\le$ C $\le$ U. Intermediate or decimal values are welcome; do not round solely to match the printed tick marks. There are no right or wrong subjective answers.

\medskip
\noindent\textit{Note.} The questionnaire was administered in Italian. The version reproduced above is an English translation prepared by the authors for this paper; the content, item order and response scales of the original form are unchanged.

\endgroup
\noindent\rule{\textwidth}{0.8pt}

\bibliographystyle{elsarticle-num-names}
\bibliography{references}

\end{document}